\documentclass[10pt, twocolumn]{IEEEtran}

\usepackage{psfrag}
\usepackage{amssymb}
\usepackage{amsmath}
\usepackage{pifont}
\usepackage{cite}
\usepackage{graphics}
\usepackage{graphicx} 
\usepackage{epsfig}
\usepackage{subfigure}
\usepackage{tablefootnote}
\usepackage{url}
\usepackage{amscd}
\usepackage{threeparttable}
\usepackage{enumerate}
\usepackage[colorlinks, citecolor=blue,linkcolor=blue]{hyperref}

\newtheorem{theorem}{Theorem}[section]
\newtheorem{remark}{Remark}
\newtheorem{lemma}[theorem]{Lemma}

\newtheorem{proposition}[theorem]{Proposition}
\newtheorem{corollary}[theorem]{Corollary}

\newenvironment{proof}[1][Proof]{\begin{trivlist}
\item[\hskip \labelsep {\bfseries #1}]}{\end{trivlist}}
\newcommand{\qed}{\nobreak \ifvmode \relax \else
      \ifdim\lastskip<1.5em \hskip-\lastskip
      \hskip1.5em plus0em minus0.5em \fi \nobreak
      \vrule height0.75em width0.5em depth0.25em\fi}

\begin{document}

\title{Support Recovery with Orthogonal Matching Pursuit in the Presence of Noise: A New Analysis}

\author{\IEEEauthorblockN{Jian Wang} \\
\IEEEauthorblockA{Department of Electrical \& Computer Engineering \\Duke University\\
Durham, North Carolina, 27708, USA \\
E-mail: jian.wang@duke.edu}}

\maketitle

\begin{abstract}
Support recovery of sparse signals from compressed linear measurements is a fundamental problem in compressed sensing (CS). In this paper, we study the orthogonal matching pursuit (OMP) algorithm for the recovery of support under noise. We consider two signal-to-noise ratio (SNR) settings: i) the SNR depends on the sparsity level $K$ of input signals, and ii) the SNR is an absolute constant independent of $K$. For the first setting, we establish necessary and sufficient conditions for the exact support recovery with OMP, expressed as lower bounds on the SNR. Our results indicate that in order to ensure the exact support recovery of all $K$-sparse signals with the OMP algorithm, the SNR must at least scale linearly with the sparsity level $K$.
In the second setting, since the necessary condition on the SNR is not fulfilled, the exact support recovery with OMP is impossible. However, our analysis shows that recovery with an arbitrarily small but constant fraction of errors is possible  with the OMP algorithm. This result may be useful for some practical applications where obtaining some large fraction of support positions is adequate.
\end{abstract}


{\IEEEpeerreviewmaketitle}

\maketitle
% make the title area
%
%
%\setcounter{page}{1}

% ============================================================== %
%
% =================== Introduction section ===================== %
%
% ============================================================== %

\section{Introduction}
We consider the support recovery of a high-dimensional sparse signal from a small number of linear measurements. This is a fundamental problem in compressed sensing (CS)~\cite{donoho2006compressed,candes2006robust} and has also received much attention in the fields of sparse approximation~\cite{devore1993constructive}, signal denoising~\cite{chen1999atomic}, and statistical model selection~\cite{schaefer1992subset}. Let $\mathbf{x} \in \mathcal{R}^n$ be a $K$-sparse signal (i.e., $\|\mathbf{x}\|_0 \leq K \ll n$) and $\mathbf{\Phi} \in \mathcal{R}^{m \times n}$ ($m < n$) be the measurement matrix. The measurements are given by
\begin{equation}
  \mathbf{y} = \mathbf{\Phi x} + \mathbf{v}
\end{equation}
where $\mathbf{v}$ is the corrupting noise. The goal of support recovery is to identify the support (i.e., the positions of nonzero elements) of the input signal $\mathbf{x}$ from the measurement vector $\mathbf{y}$. It is known that optimal support recovery requires an exhaustive search over all possible support sets of the sparse signal and hence is NP-hard~\cite{natarajan1995sparse}. For this reason, much attention has been drawn to computationally efficient approaches. In this paper we consider the orthogonal matching pursuit (OMP) algorithm for solving the support recovery problem.  OMP is a canonical greedy algorithm for sparse approximation in signal processing~\cite{pati1993orthogonal,mallat1993matching}. It is also known as greedy least square regression in statistics~\cite{schaefer1992subset} and forward greedy selection in machine learning~\cite{zhang2009consistency}. The principle of the OMP algorithm is quite simple: it iteratively identifies the support of the sparse signal, by adding one index into the list at a time according to the maximum correlation between columns of measurement matrix and the current residual. There are several popular stopping rules for the OMP algorithm that can be implemented at minimal cost~\cite{tropp2010computational}:
\begin{enumerate}[i)]

\item Stop after a fixed number of iterations: $k = K$.

\item Stop when the energy in the residual is small: $\|\mathbf{r}^k \|_2\leq \epsilon$. 

\item Stop when no column in the measurement matrix is strongly correlated with the  residual: $\|\mathbf{\Phi}' \mathbf{r}^{k - 1} \|_\infty \leq \epsilon$.

\end{enumerate}
See Table \ref{tab:omp} for the mathematical description of a version of OMP. Both in theory and in practice, the OMP algorithm has demonstrated competitive performance~\cite{tropp2007signal}.

Over the years, many efforts have been made to analyze the performance of OMP in sparse support recovery. In one line of work, probabilistic analyses have been proposed. Tropp and Gilbert showed that when the measurement matrix $\mathbf{\Phi}$ is generated iid at random, OMP can ensure the accurate recovery of every fixed $K$-sparse signal from noise-free measurements with overwhelming probability with~\cite{tropp2007signal} \begin{equation}
m \geq cK \log n
\end{equation} for some constant $c$.
Fletcher and Rangan provided an improved scaling law on the number of measurements and also showed that the scaling law works for noisy scenarios for which the signal-to-noise (SNR) goes to infinity~\cite{fletcher2012orthogonal}.

Another direction is to develop deterministic conditions for the exact support recovery with the OMP algorithm~\cite{tropp2004greed,davenport2010analysis,liu2012orthogonal,huang2011recovery,zhang2011sparse,livshits2012efficiency,wang2012Recovery,mo2012remarks,soussen2013joint}. Those conditions are often characterized by the properties of measurement matrices, such as the mutual incoherence property (MIP)~\cite{donoho2001uncertainty} and the restricted isometry property (RIP)~\cite{candes2005decoding}. A measurement matrix $\mathbf{\Phi}$ is said to satisfy
the RIP of order $K$ if there exists a constant $c \in [0, 1)$
such that
\begin{equation}
  \label{eq:RIP} (1- c) \| \mathbf{x} \|_2^2 \leq \| \mathbf{\Phi x} \|_2^2 \leq (1 + c) \| \mathbf{x} \|_2^2
\end{equation}
for all $K$-sparse vectors $\mathbf{x}$. In particular, the minimum value among all constants $c$ satisfying (\ref{eq:RIP}) is called the isometry constant $\delta_{K}$. In the noise-free case (i.e., when the noise vector $\mathbf{v} = \mathbf{0}$), Davenport and Wakin showed that~\cite{davenport2010analysis}
\begin{equation}
\delta_{K + 1} < \frac{1}{3\sqrt K}
\end{equation} is sufficient for OMP to accurately recover the support of the input signal. For further improvements on this condition, see~\cite{liu2012orthogonal,huang2011recovery,zhang2011sparse,livshits2012efficiency,wang2012Recovery,mo2012remarks,wen2013improved,chang2014improved}. The deterministic conditions on the exact support recovery with OMP in the noisy case have been studied in~\cite{zhang2009consistency,cai2011orthogonal,shen2011sparse,wu2013exact,chang2014improved}, in which the researchers considered the OMP algorithm with residual-based stopping rules and established sufficient conditions for the exact support recovery that depend on the properties of measurement matrices and the minimum magnitude of the nonzero elements of the signal.

\setlength{\arrayrulewidth}{1.5pt}
\begin{table} 
  \centering
      \caption[]{%
      The OMP Algorithm \label{tab:omp}}
      \vspace{-1mm}
    %\begin{footnotesize}
\begin{tabular}{@{}ll}
\hline \\ \vspace{-12pt} \\
\textbf{Input}       &$\mathbf{\Phi}$, $\mathbf{y}$, and sparsity level $K$. \\
\textbf{Initialize}  & iteration counter $k = 0$, \\
                     & estimated support $\mathcal{T}^{0} = \emptyset$, \\
                     & and residual vector $\mathbf{r}^{0} = \mathbf{y}$.
                     \\
%\mr
\textbf{While}       & $k < K$ \textbf{do}\\
                     & $k = k + 1$. \\
                     & Identify\tablefootnote{There is another popular version of OMP in the literature which normalizes columns of $\mathbf{\Phi}$ before the identification step (see, e.g.,~\cite{fletcher2012orthogonal}). In the present paper we exclusively consider the version of OMP without column normalization, since results obtained for this version can be readily extended to the version with normalization via rescaling. Note that when $\mathbf{\Phi}$ has normalized columns (e.g., dictionary atoms), the two versions of OMP become identical.} \hspace{0.439mm}${t}^{k}  = \underset{i \in \Omega \backslash \mathcal{T}^{k - 1}}{\arg \max} |\langle \phi_i, \mathbf{r}^{k - 1} \rangle|$. \\
                     & Enlarge \hspace{0.58mm}$\mathcal{T}^{k} = \mathcal{T}^{k - 1} \cup t^{k}$. \\
                     & Estimate $\mathbf{x}^{k} = \underset{\mathbf{u}:\textit{supp}(\mathbf{u}) = \mathcal{T}^k}{\arg \min} \|\mathbf{y}-\mathbf{\Phi} \mathbf{u}\|_2$. \\
                     & Update \hspace{2.1mm}$\mathbf{r}^{k} = \mathbf{y} - \mathbf{\Phi} \mathbf{x}^{k}$. \\
\textbf{End}         \\
\textbf{Output}      & the estimated support $\mathcal{T}^{K}$ and vector $\mathbf{x}^K$.
\vspace{4pt} \\
\hline
\end{tabular}  
\end{table}
\setlength{\arrayrulewidth}{1.3pt}

The main purpose of this paper is to investigate deterministic conditions of OMP for the support recovery in the noisy case. Unlike previous studies that considered residual-based stopping rules for the OMP algorithm, we simply consider that OMP runs $K$ iterations before stopping, which is arguably the most natural stopping rule if one is concerned with the recovery of exact support. We establish necessary and sufficient conditions for the exact support recovery with OMP, expressed as lower bounds on the SNR. Our results indicate that the OMP algorithm can accurately recover the support of all $K$-sparse signals only when the SNR scales at least linearly with the sparsity $K$. For high-dimensional setting, this essentially requires the SNR to be unbounded from above.

We also study the situation where the SNR is upper bounded so that the necessary condition for the exact support recovery with OMP is not fulfilled. The analysis of OMP with bounded SNR has been an interesting open problem~\cite{fletcher2012orthogonal}.
We consider a practical case where the SNR is an absolute constant independent of the sparsity $K$. Our result shows that under appropriate conditions on the SNR and the isometry constant, OMP can approximately recover the support of sparse signals with only a small constant fraction of errors.

The main contributions of this paper are summarized as follows.
\begin{enumerate}[i)]
\item
We consider the exact support recovery with OMP in the noisy scenario. Our analysis shows that OMP can accurately recover the support of any $K$-sparse signal if
  \begin{equation}
\sqrt{\text{SNR}} > \frac{2 \sqrt{K} (1 + \delta_{K + 1})}{(1 - (\sqrt{K} + 1) \delta_{K + 1}) \cdot \sqrt{\text{MAR}}}, \label{eq:jjjjffffa00}
  \end{equation}
where MAR is the minimum-to-average ratio of the input signal (see definitions of the MAR and the SNR in Section~\ref{sec:III}). We also establish a necessary condition for the exact support recovery with OMP as
    \begin{equation}
\sqrt{\text{SNR}} > \frac{\sqrt{K} (1 + \delta_{K + 1})  }{( 1 - \sqrt{K} \delta_{K + 1}) \cdot \sqrt{\text{MAR}}}. \label{eq:jjjjffffa01}
  \end{equation}
Therefore, to ensure the perfect support recovery of all $K$-sparse signals with OMP, the SNR must at least scale linearly with the sparsity level $K$ of input signals. 
\vspace{2mm}

\item
We also consider the case where $\text{SNR}$ is an absolute constant independent of the signal sparsity $K$. Our analysis shows that if
\begin{equation}
 \sqrt{\text{SNR}} \geq \frac{\kappa}{\delta_{2K}^{3/4}}
 \end{equation} where $\kappa := \max_{i,j \in \textit{supp}(\mathbf{x})} \frac{|x_{i}|}{|x_{j}|}$, then OMP can recover the support of any $K$-sparse signal with error rate
\begin{equation}
\rho_{\text{error}} \leq C \kappa^2 \delta_{2K}^{1/2},
\end{equation}
where $C$ is a constant. Therefore, with properly chosen isometry constants, the fraction of errors in the support recovery can be made arbitrarily small.
\end{enumerate}
\vspace{2mm}

The rest of this paper is organized as follows: In Section~\ref{sec:II}, we introduce notations and lemmas that are used in this paper. In Section~\ref{sec:III}, we analyze necessary and sufficient conditions for the exact support recovery with OMP. In Section~\ref{sec:VI}, we provide the results of OMP for the approximate support recovery of sparse signals. Concluding remarks are given in Section~\ref{sec:V}.

% ============================================================== %
%
% =================== Preliminaries section ==================== %
%
% ============================================================== %

\section{Preliminaries} \label{sec:II}

\subsection{Notations}
We briefly summarize notations used in this paper. Let $\Omega
= \{1,2, \cdots ,n\}$ and $\mathcal{T}= \textit{supp}(\mathbf{x}) =
\{i|i \in \Omega,x_{i} \neq 0\}$ denote the support of vector
$\mathbf{x}$. For $\mathcal{S} \subseteq \Omega$,
$|\mathcal{S}|$ is the cardinality of $\mathcal{S}$. $\mathcal{T}
\setminus \mathcal{S}$ is the set of all elements contained in
$\mathcal{T}$ but not in $\mathcal{S}$. $\mathbf{x}_{\mathcal{S}}
\in \mathcal{R}^{|\mathcal{S}|}$ represents a restriction of the vector
$\mathbf{x}$ to the elements with indices in $\mathcal{S}$.
$\mathbf{\Phi}_{\mathcal{S}} \in \mathcal{R}^{m \times | \mathcal{S}
|}$ is a submatrix of $\mathbf{\Phi}$ that only contains columns
indexed by $\mathcal{S}$. If $\mathbf{\Phi}_{\mathcal{S}}$ is full
column rank, then $\mathbf{\Phi}_{\mathcal{S}}^{\dagger} = (
\mathbf{\Phi}'_{\mathcal{S}} \mathbf{\Phi}_{\mathcal{S}} )^{-1}
\mathbf{\Phi}'_{\mathcal{S}}$ is the pseudoinverse of
$\mathbf{\Phi}_{\mathcal{S}}$. $\text{span} (
\mathbf{\Phi}_{\mathcal{S}} )$ represents the span of columns in
$\mathbf{\Phi}_{\mathcal{S}}$. $\mathbf{P}_{\mathcal{S}} =
\mathbf{\Phi}_{\mathcal{S}} \mathbf{\Phi}_{\mathcal{S}}^{\dagger}$
stands for the projection onto $\text{span} (
\mathbf{\Phi}_{\mathcal{S}} )$. $\mathbf{P}_{\mathcal{S}}^{\bot} =
\mathbf{I} - \mathbf{P}_{\mathcal{S}}$ is the projection onto the
orthogonal complement of $\text{span} (
\mathbf{\Phi}_{\mathcal{S}} )$, where $\mathbf{I}$ denotes the
identity matrix.

\subsection{Lemmas}
The following lemmas are useful for our analysis.
\vspace{2mm}

\begin{lemma}
  [Lemma 3 in {\cite{candes2005decoding}}]\label{lem:mono}If a measurement
  matrix satisfies the RIP of both orders $K_{1}$ and $K_{2}$ where $K_{1} \leq K_{2}$, then
  $$\delta_{K_{1}} \leq \delta_{K_{2}}.$$ This
  property is often referred to as the monotonicity of the isometry constant.
\end{lemma}
\vspace{2mm}

\begin{lemma}
  [Direct consequences of RIP {\cite{needell2009cosamp,kwon2013multipath}}]\label{lem:rips} Let $\mathcal{S}
  \subseteq \Omega$. If $\delta_{| \mathcal{S} |} <1$ then for any $\mathbf{u} \in
  \mathcal{R}^{| \mathcal{S} |}$,
  \begin{eqnarray}
    ( 1- \delta_{| \mathcal{S} |} ) \left\| \mathbf{u} \right\|_{2} \leq \left\|
    \mathbf{\Phi}_{\mathcal{S}}'  \mathbf{\Phi}_{\mathcal{S}} \mathbf{u} \right\|_{2} \leq ( 1+
    \delta_{| \mathcal{S} |} ) \left\| \mathbf{u} \right\|_{2}, \nonumber\\
    \frac{1}{1+ \delta_{| \mathcal{S} |}} \left\| \mathbf{u} \right\|_{2} \leq \| (
    \mathbf{\Phi}_{\mathcal{S}}'  \mathbf{\Phi}_{\mathcal{S}} )^{-1} \mathbf{u} \|_{2} \leq
    \frac{1}{1- \delta_{| \mathcal{S} |}} \left\| \mathbf{u} \right\|_{2}.
    \nonumber 
  \end{eqnarray}
\end{lemma}
\vspace{2mm}

\begin{lemma}
  [Consequences of RIP \cite{candes2008restricted,shen2014analysis}] \label{lem:correlationrip} Let
  $\mathcal{S}_{1},\mathcal{S}_{2},\mathcal{S}_{3} \subseteq \Omega$ and $\mathcal{S}_{1} \cap \mathcal{S}_{2} \cap \mathcal{S}_{2} =
  \emptyset$. If $\delta_{|\mathcal{S}_{1} | + |\mathcal{S}_{2} | + |\mathcal{S}_{3} |} <1$, then for any vector $\mathbf{v} \in \mathcal{R}^{|\mathcal{S}_{2}|}$,
\begin{eqnarray}
  \| \mathbf{\Phi}_{\mathcal{S}_{1}}'  \mathbf{\Phi}_{\mathcal{S}_2} \mathbf{v} \|_{2} &\leq& \delta_{|\mathcal{S}_{1} | +
     |\mathcal{S}_{2} |} \left\| \mathbf{v} \right\|_{2}, \nonumber \\
    \| \mathbf{\Phi}_{\mathcal{S}_{1}}' \mathbf{P}^\bot_{\mathcal{S}_3}  \mathbf{\Phi}_{\mathcal{S}_2} \mathbf{v}  \|_{2} &\leq& \delta_{|\mathcal{S}_{1} | + |\mathcal{S}_{2}| + |\mathcal{S}_{3} |} \left\| \mathbf{v} \right\|_{2}. \nonumber
\end{eqnarray}
\end{lemma}
\vspace{2mm}

\begin{lemma}[Proposition 3.1 in~\cite{needell2009cosamp}] \label{lem:rip5}
Let $\mathcal{S} \subseteq \Omega$. If $\delta_{|\mathcal{S} |} < 1$, then for any vector $\mathbf{u} \in \mathcal{R}^m$,
\begin{equation}
\|\mathbf{\Phi}'_\mathcal{S} \mathbf{u}\|_2 \leq \sqrt{1 + \delta_{|\mathcal{S}|}} \|\mathbf{u}\|_2. \nonumber
\end{equation}
\end{lemma}
\vspace{2mm}

\begin{lemma}[Lemma 5 in~\cite{cai2011orthogonal}] \label{lem:rip6} Let
  $\mathcal{S}_{1},\mathcal{S}_{2}, \subseteq \Omega$ be disjoint sets and denote $\mathbf{A} = \mathbf{\Phi}'_{\mathcal{S}_1}  \mathbf{P}^\bot_{\mathcal{S}_2}  \mathbf{\Phi}_{\mathcal{S}_1 }$ and $\mathbf{B} = \mathbf{\Phi}'_{\mathcal{S}_1 \cup \mathcal{S}_2} \mathbf{\Phi}_{\mathcal{S}_1 \cup \mathcal{S}_2}$.  Then the minimum and maximum eigenvalues of $\mathbf{A}$ and $\mathbf{B}$ satisfy
\begin{equation}
    \lambda_{\min} ( \mathbf{A}) \geq \lambda_{\min} (\mathbf{B})~~\text{and}~~ \lambda_{\max} (\mathbf{A}) \leq \lambda_{\max} (\mathbf{B}). \nonumber
  \end{equation}
\end{lemma}

% ============================================================== %
%
% ======================= Analysis section ===================== %
%
% ============================================================== %

\section{Exact Support Recovery via OMP in the Presence of Noise} \label{sec:III}

\subsection{Main Results}
In this section, we analyze the condition for the exact support recovery with OMP in the presence of noise. Following the analysis in~\cite{fletcher2012orthogonal}, we parameterize the dependence on the noise $\mathbf{v}$ and the signal $\mathbf{x}$ with two quantities: $$\text{SNR}:= \frac{\|\mathbf{\Phi x}\|_2^2}{\|\mathbf{v}\|_2^2}.$$
and
$$\text{MAR} = \frac{\min_{j \in \mathcal{T}} |x_{j}|^2}{\|\mathbf{x}\|_2^2 /{K}}.$$
The next theorem provides a condition under which OMP can accurately recover the support of any $K$-sparse signal $\mathbf{x}$. 
\vspace{2mm}

\begin{theorem} [Sufficient Condition] \label{thm:2}
Suppose that the measurement matrix $\mathbf{\Phi}$ satisfies the RIP with $\delta_{K + 1} < \frac{1}{\sqrt K + 1}$. Then OMP performs the exact support recovery of any $K$-sparse signal $\mathbf{x}$ from its noisy measurements $\mathbf{y} = \mathbf{\Phi x} + \mathbf{v}$, provided that
  \begin{equation}
\sqrt{\text{SNR}} > \frac{2 \sqrt K (1 + \delta_{K + 1})  }{( 1 - (\sqrt{K} + 1) \delta_{K + 1} )  \cdot \sqrt{\text{MAR}}}.  \label{eq:jjjjffffa33}
  \end{equation}
\end{theorem}
\vspace{2mm}

One can interpret from Theorem~\ref{thm:2} that in the high-dimensional setting, the exact support recovery with OMP can be ensured in the high SNR region. % (i.e., $\text{SNR} = \mathcal{O}(K)$).
Moreover, observe that \eqref{eq:jjjjffffa33} can be rewritten as
    \begin{equation}
    \delta_{K + 1} <  \frac{1 - t}{\sqrt K + 1 + t}, \label{eq:relaxation122}
    \end{equation}
    where $$t := \frac{2 \sqrt K}{\sqrt{\text{SNR} \cdot \text{MAR}}}.$$
    Hence, when $t \rightarrow 0$, the condition reduces to
\begin{equation}
\delta_{K + 1} < \frac{1}{\sqrt K + 1}, \label{eq:goodcond}
\end{equation}
which coincides with the recovery condition of OMP in the noise-free case~\cite{wang2012Recovery,mo2012remarks}. The condition in \eqref{eq:goodcond} has also been shown to be nearly necessary for the exact support recovery with OMP since it cannot be further relaxed to $\delta_{K + 1} \leq \frac{1}{\sqrt K }$~\cite{herzet2012exact,wang2012Recovery,mo2012remarks}.
\vspace{2mm}

The following theorem gives a necessary condition for the exact support recovery with the OMP algorithm. 
\vspace{2mm}

\begin{theorem} [Necessary Condition] \label{thm:3}
If one wish to accurately recover the support of any $K$-sparse signal $\mathbf{x}$ from its noisy measurements $\mathbf{y} = \mathbf{\Phi x} + \mathbf{v}$ with OMP, then the SNR should satisfy
\begin{equation}
\sqrt{\text{SNR}} > \frac{\sqrt{K}(1 + \delta_{K + 1})}{( 1 - \sqrt{K} \delta_{K + 1} ) \cdot \sqrt{\text{MAR}}}.  \label{eq:jjjjffffa}
  \end{equation}
\end{theorem}
\begin{proof}
See Appendix \ref{app:thm3}.
\end{proof}
\vspace{2mm}

\begin{remark}

\

\begin{enumerate}[i)]
\item
Loosely speaking, the lower bounds in \eqref{eq:jjjjffffa33} and \eqref{eq:jjjjffffa} can be matched within a constant factor of two. Interestingly, since ${\text{MAR}} \leq 1$ and $\delta_{K + 1} \geq 0$, one can directly obtain from  \eqref{eq:jjjjffffa} that \begin{equation}
\text{SNR} > K,
\end{equation} which implies that in order for OMP to accurately recover the support of any $K$-sparse signal, the SNR must at least scale linearly with the sparsity level $K$ of the signal. For high-dimensional setting, this essentially requires the SNR to be unbounded from above.
\vspace{2mm}

\item
The sufficient condition for the exact support recovery with OMP under a similar SNR setting has been studied in~\cite{fletcher2012orthogonal}, in which the authors considered large random measurement matrices. They provided asymptotic and probabilistic results on the scaling law for the number of measurements that ensures the exact support recovery of every fixed $K$-sparse signal when the SNR approaches to infinity. Our result in Theorem~\ref{thm:2} extends that in~\cite{fletcher2012orthogonal} by providing a deterministic condition that applies to general measurement matrices and also holds uniformly for all $K$-sparse signals.
\vspace{2mm}

\item
The result in Theorem~\ref{thm:2} is also closely related to the results  in~\cite{zhang2009consistency,cai2011orthogonal,shen2011sparse,wu2013exact,chang2014improved}, in which the authors considered the OMP algorithm with residual-based stopping rules by assuming that the noise level is known a priori. They established conditions that depend on the properties of measurement matrices and the minimum magnitude of nonzero elements in the signal $\mathbf{x}$. In fact, it can be shown that the conditions proposed in~\cite{zhang2009consistency,cai2011orthogonal,shen2011sparse,wu2013exact,chang2014improved} essentially impose a similar requirement on the size of measurements as our result established in Theorem~\ref{thm:2}. However, our analysis differs in 1) that we characterize the sufficient condition for the exact support recovery with a lower bound on the SNR and 2) that we consider the stopping rule that the OMP algorithm runs exact $K$ iterations before stopping and does not require the assumption of knowing the noise level. In addition, we provide in Theorem \ref{thm:3} the necessary condition analysis for the exact support recovery with OMP, for which there is no counterpart in the studies of~\cite{zhang2009consistency,cai2011orthogonal,shen2011sparse,wu2013exact,chang2014improved}.
\vspace{2mm}

\item
For the OMP algorithm, it had become usual, when recovering a $K$-sparse signal, to consider the performance of the algorithm after $K$ iterations. See, for instance,~\cite{tropp2004greed,tropp2007signal,davenport2010analysis,liu2012orthogonal,fletcher2012orthogonal,wang2012Recovery,mo2012remarks,chang2014improved}. Compared to stopping rules that are based on the residual tolerance (e.g.,~\cite{cai2011orthogonal}) or the number of maximally allowed iterations (e.g., $\lceil cK \rceil$ iterations, $c > 1$~\cite{zhang2011sparse,livshits2012efficiency,foucart2013stability,wang2012how,livshitz2014sparse}), running OMP for $K$ iterations is very natural in that it directly allows the algorithm to recover the exact support set without false alarms or missed detections since the $K$-th iterate is itself a $K$-sparse signal. This feature is appealing when one is concerned with the exact support recovery. 
\vspace{2mm}

\item
It is worth mentioning that for the noise-free case, the stopping rule of OMP does not require additional information about the data, since it can simply be $\|\mathbf{r}^k\|_2 = 0$ (or $\|\mathbf{\Phi}' \mathbf{r}^{k - 1}\|_\infty = 0$).\footnote{Without noise, the stopping rule of $\|\mathbf{r}^k\|_2 = 0$ (or $\|\mathbf{\Phi}' \mathbf{r}^{k - 1}\|_\infty = 0$) is essentially equivalent to running OMP for $K$ iterations.}
Whereas, in the noisy case, some prior information on the signal and the noise is often needed. For example, residual-based stopping rules may rely on the prior knowledge of the noise level~\cite{zhang2009consistency,cai2011orthogonal,shen2011sparse,wu2013exact,chang2014improved}, and running OMP for $K$ iterations requires to know a priori the sparsity level $K$ of the input signal. The assumption of knowing the sparsity (or its region) has been commonly made for the algorithm design and the performance analysis in the CS field (see, e.g.,~\cite{needell2010signal,donoho2006sparse,needell2009cosamp,dai2009subspace,blumensath2009iterative,foucart2011hard}). However, it should be noted that in many applications, the sparsity $K$ is often not available and moreover, the underlying signal may not be exactly sparse. A typical scenario is that the input signal is not exactly sparse but only approximately sparse  with a few significant nonzero coefficients.\footnote{Note that if the signal to be recovered is not even approximately sparse, then compressed sensing technique may not apply.} For this scenario, our analysis for the recovery of the exactly sparse signal can be readily extended to the recovery of significant nonzero coefficients of the signal, by treating the contribution of those small nonzero coefficients as part of noise and utilizing the techniques developed in~\cite{foucart2011hard,cai2011orthogonal}.
\vspace{2mm}

\end{enumerate}

\

\end{remark}

\subsection{Proof of Theorem~\ref{thm:2}}
Our proof of Theorem~\ref{thm:2} is essentially an extension of the proof technique in \cite[Theorem 3.4 and 3.5]{wang2012Generalized}. Note that \cite[Theorem 3.4 and 3.5]{wang2012Generalized} studied the recovery condition for the generalized OMP (gOMP) algorithm in the noise-free situation. Our contribution is to generalize the analysis to the noisy case. We would like to mention that~\cite{wang2012Generalized} also provided a noisy case analysis but focused only on the $\ell_2$-norm distortion of the signal recovery and the corresponding result is also weaker than the result established in this paper (see Section~\ref{sec:dist}).

The proof works by mathematical induction. For the convenience of stating the results, we say that OMP makes a success at an iteration if it selects a correct index at the iteration. We will first give a condition that guarantees the success of OMP at the first iteration. Then we will assume that OMP has been successful in the previous $k$ ($1 \leq k < K$) iterations and will derive a condition under which OMP also makes a success at the $(k + 1)$-th iteration. Finally, we will combine the two conditions to establish an overall condition for the OMP algorithm.

\begin{itemize}
\item 
\textbf{Success at the first iteration}:
\vspace{1mm}

From Table~\ref{tab:omp}, we know that at the first iteration, OMP selects the index $t^1$ corresponding to the column ${\phi}_{t^{1}} \in \mathbf{\Phi}$ that is most strongly correlated with the measurement vector $\mathbf{y}$. Hence,
  \begin{eqnarray}
    |\langle {\phi}_{t^{1}}, \mathbf{y} \rangle| \hspace{-2mm}& = & \hspace{-2mm} \max_{i \in \Omega} |\langle {\phi}_{i}, \mathbf{y} \rangle| \nonumber \\
    \hspace{-2mm} & \geq & \hspace{-2mm}\max_{i \in \mathcal{T}} |\langle {\phi}_{i}, \mathbf{y} \rangle| \nonumber \\
    \hspace{-2mm}
    & \geq & \hspace{-2mm}
    \sqrt{\frac{1}{|\mathcal{T}|} \sum_{i \in \mathcal{T}} \langle {\phi}_{i}, \mathbf{y} \rangle^2 } \nonumber \\
    \hspace{-2mm} & = & \hspace{-2mm}
    \frac{1}{\sqrt K} \| \mathbf{\Phi}_{\mathcal{T}}' \mathbf{y} \|_{2}
    \nonumber \\
    \hspace{-2mm} & \overset{(a)}{=} & \hspace{-2mm}
    \frac{1}{\sqrt K} \| \mathbf{\Phi}_{\mathcal{T}}' \mathbf{\Phi}_{\mathcal{T}} \mathbf{x}_{\mathcal{T}} + \mathbf{\Phi}_{\mathcal{T}}' \mathbf{v} \|_{2}
    \nonumber\\
    \hspace{-2mm}  & \overset{(b)}{\geq} & \hspace{-2mm} \frac{1}{\sqrt K} \left(\| \mathbf{\Phi}_{\mathcal{T}}'  \mathbf{\Phi}_{\mathcal{T}}
    \mathbf{x}_{\mathcal{T}} \|_{2} - \|\mathbf{\Phi}_{\mathcal{T}}' \mathbf{v}\|_{2} \right) \nonumber\\
     \hspace{-2mm}& \overset{(c)}{\geq} & \hspace{-2mm} \frac{1}{\sqrt K} \left((1- \delta_{K} ) \| \mathbf{x} \|_{2} - \|\mathbf{\Phi}_{\mathcal{T}}' \mathbf{v}\|_{2} \right) \nonumber \\
    \hspace{-2mm} & \overset{(d)}{\geq} & \hspace{-2mm} \frac{1}{\sqrt K} \left((1- \delta_{K} ) \| \mathbf{x} \|_{2} - \sqrt{1 + \delta_K}~\|\mathbf{v}\|_{2} \right), \nonumber \\
    \label{eq:try5}
  \end{eqnarray}
where (a) is because $\mathbf{y} = \mathbf{\Phi x} + \mathbf{v}$, (b) is from the triangle inequality, (c) is from the RIP, and (d) is due to Lemma~\ref{lem:rip5}.
\vspace{2mm}

On the other hand, if a wrong index is chosen at the first iteration (i.e., $t^{1} \notin \mathcal{T}$), then
  \begin{eqnarray}
    |\langle \phi_{t^1}, \mathbf{y} \rangle|
    &=& |\langle \phi_{t^1}, \mathbf{\Phi x} \rangle + \langle \phi_{t^1}, \mathbf{v} \rangle| \nonumber \\
    &\overset{(a)}{\leq}& |\langle \phi_{t^1}, \mathbf{\Phi x} \rangle| + |\langle \phi_{t^1}, \mathbf{v} \rangle| \nonumber \\
    &=& \| \phi'_{t^1} \mathbf{\Phi}_{\mathcal{T}} \mathbf{x}_{\mathcal{T}} \|_2 + \|\phi'_{t^1} \mathbf{v} \|_{2} \nonumber \\
    &\overset{(b)}{\leq}& \delta_{K+1} \left\| \mathbf{x} \right\|_{2} + \|\phi'_{t^1} \mathbf{v} \|_{2} \nonumber \\
    &\overset{(c)}{\leq}& \delta_{K+1} \left\| \mathbf{x} \right\|_{2} + \sqrt{1 + \delta_1} \|\mathbf{v}\|_{2}, \label{eq:365}
  \end{eqnarray}
where (a) uses the triangle inequality, (b) follows from Lemma~\ref{lem:correlationrip} and  (c) is due to Lemma~\ref{lem:rip5}. 
\vspace{2mm}

This, however, contradicts {\eqref{eq:try5}} whenever
  \begin{eqnarray}
      \lefteqn{\delta_{K + 1} \| \mathbf{x} \|_{2} + \sqrt{1 + \delta_1} ~\|\mathbf{v}\|_{2}} \nonumber \\
      && < \frac{1}{\sqrt{K}} \left((1- \delta_{K} ) \| \mathbf{x} \|_{2} - \sqrt{1 + \delta_K} ~\|\mathbf{v}\|_{2} \right).~~~\label{eq:k+1n440}
  \end{eqnarray}

Since $\delta_1 \leq \delta_K \leq \delta_{K + 1}$ (by the monotonicity of isometry constant), \eqref{eq:k+1n440} is guaranteed by
  \begin{eqnarray}
      \lefteqn{\delta_{K + 1} \| \mathbf{x} \|_{2} + \sqrt{1 + \delta_{K + 1}} \|\mathbf{v}\|_{2}} \nonumber \\
      && \hspace{-4mm} < \frac{1}{\sqrt{K}} \left((1- \delta_{K + 1} ) \| \mathbf{x} \|_{2} - \sqrt{1 + \delta_{K + 1}} \| \mathbf{v}\|_{2} \right). ~~~~~ \label{eq:15ge}
  \end{eqnarray}
Or equivalently,
  \begin{eqnarray}
      \lefteqn{\left( \frac{1 - \delta_{K + 1}}{\sqrt{K}} - \delta_{K + 1} \right) \|\mathbf{x}\|_2} \nonumber \\
      &&> \left(1 + \frac{1}{\sqrt{K}} \right) \sqrt{1 + \delta_{K + 1}} \|\mathbf{v}\|_2.  \label{eq:11000121}
  \end{eqnarray}
Furthermore, since
\begin{equation}
 \|\mathbf{\Phi x}\|_2 \leq \sqrt{1 + \delta_{K}} \|\mathbf{x}\|_2 \leq \sqrt{1 + \delta_{K + 1}} \|\mathbf{x}\|_2, \label{eq:RI}
 \end{equation}
and also noting that $\text{SNR} = \frac{\|\mathbf{\Phi x}\|_2^2}{\|\mathbf{v}\|_2^2}$, we can show that \eqref{eq:11000121} holds true if
\begin{equation}
 \frac{1 - \delta_{K + 1}}{\sqrt{K}} - \delta_{K + 1}   > \frac{1 + \delta_{K + 1}}{\sqrt{\text{SNR}}}\left( 1 + \frac{1}{\sqrt{K}} \right). \label{eq:firstite0}
\end{equation}
That is,
\begin{equation}
\sqrt{\text{SNR}} > \frac{(\sqrt K + 1) (1 + \delta_{K + 1}) }{1 - (\sqrt K + 1) \delta_{K + 1}}, \label{eq:firstite}
\end{equation}

Therefore, under \eqref{eq:firstite}, a correct index is chosen at the first iteration of OMP.
\vspace{3mm}

\item 
\textbf{Success at the general iteration}:
\vspace{1mm}

Assume that OMP has been successful in each of the previous $k$ ($1 \leq k < K$) iterations. Then,
  \begin{equation} \label{eq:200002}
    |\mathcal{T} \cap \mathcal{T}^{k}| = k.
  \end{equation}
Under this assumption, we will derive a condition that ensures OMP to make a success at the $(k+1)$-th iteration as well. 
\vspace{1mm}

For analytical convenience, we introduce two quantities. Let $u$ denote the largest value in $\{ |\langle \phi_{i}, \mathbf{r}^{k} \rangle| \}_{i \in \mathcal{T} \backslash \mathcal{T}^k}$ and let $v$ denote the largest value in $\{ |\langle \phi_{i}, \mathbf{r}^{k} \rangle| \}_{i \in \Omega \setminus (\mathcal{T} \cup \mathcal{T}^{k} )}$. Note that $\mathcal{T} \backslash \mathcal{T}^k$ and $\Omega \setminus (\mathcal{T} \cup \mathcal{T}^{k} )$ are the set of remaining correct indices and the set of remaining incorrect indices, respectively. Then it is clear that if
  \begin{equation}
    u > v,
  \end{equation}
a good index will be selected at the $(k+1)$-th iteration of OMP. The following proposition characterizes the lower bound of $u$ and the upper bound of $v$.
\vspace{1mm}

\begin{proposition}
    \label{prop:upperbound15}%It satisfies that
    \begin{eqnarray}
 u \hspace{-2mm} &\geq& \hspace{-2mm} \frac{1}{\sqrt{K - k}} \left( (1 - \delta_{K}) \| \mathbf{x}_{\mathcal{T} \backslash \mathcal{T}^k} \|_{2} - \sqrt{1 + \delta_{K}} \| \mathbf{v} \|_{2} \right), \nonumber \label{eq:small5}  \\
      v \hspace{-2mm} & \leq & \hspace{-2mm} \delta_{K + 1}
       \| \mathbf{x}_{\mathcal{T} \backslash \mathcal{T}^k} \|_{2} + \sqrt{1 + \delta_{1}} ~ \|\mathbf{v}\|_2. \nonumber
      \label{eq:large5}
    \end{eqnarray}
  \end{proposition}

\begin{proof}
See Appendix \ref{app:upperbound1}.
\end{proof}

Using Proposition \ref{prop:upperbound15}, we obtain the sufficient condition to $u > v$ as
  \begin{eqnarray}
    \label{eq:sufficientommp45}
    \lefteqn{\hspace{-6mm}\frac{1}{\sqrt{K - k}} \left((1- \delta_{K}) \| \mathbf{x}_{\mathcal{T} \backslash \mathcal{T}^k} \|_{2} - \sqrt{1 + \delta_{K}} ~\|\mathbf{v}\|_2\right)} \nonumber\\
    && > {\delta_{K + 1}} \| \mathbf{x}_{\mathcal{T} \backslash \mathcal{T}^k} \|_{2} + \sqrt{1 + \delta_{1}}~ \|\mathbf{v}\|_2, \nonumber
  \end{eqnarray}
or equivalently,
  \begin{eqnarray}
      \lefteqn{\left( \frac{1 - \delta_{K + 1}}{\sqrt{K - k}} - \delta_{K + 1} \right) \|\mathbf{x}_{\mathcal{T} \backslash \mathcal{T}^k}\|_2} \nonumber \\
      &&> \left(1 + \frac{1}{\sqrt{K - k}} \right) \sqrt{1 + \delta_{K + 1}} \|\mathbf{v}\|_2.  \label{eq:110001212}
  \end{eqnarray}
Note that
\begin{eqnarray}
\lefteqn{\|\mathbf{x}_{\mathcal{T} \backslash \mathcal{T}^k}\| } \nonumber \\
 &\geq& \sqrt{|\mathcal{T} \backslash \mathcal{T}^k|} \min_{j \in \mathcal{T}} |x_j| \nonumber\\
&\overset{(a)}{\geq}&  \sqrt{K - k} \left(\frac{ \sqrt{\text{MAR}} \cdot \|\mathbf{x}\|_2}{\sqrt{K}} \right) \nonumber\\
&\overset{(b)}{\geq}& \sqrt{\frac{K - k}{K(1 + \delta_{K + 1})}} \cdot \sqrt{\text{MAR}} \cdot \|\mathbf{\Phi x}\|_2 \nonumber \\
&\overset{(c)}{=}& \sqrt{\frac{K - k}{K(1 + \delta_{K + 1})}} \cdot \sqrt{\text{MAR} \cdot \text{SNR}} \cdot \|\mathbf{v}\|_2,~~~\label{eq:firstite02}
\end{eqnarray}
where (a) is from the definition of MAR, (b) is from \eqref{eq:RI}, and (c) is because $\text{SNR} = \frac{\|\mathbf{\Phi x}\|_2^2}{\|\mathbf{v}\|_2^2}$. 
\vspace{2mm}

Using \eqref{eq:110001212} and \eqref{eq:firstite02}, we can show $u > v$ holds true if
  \begin{eqnarray}
      \lefteqn{\hspace{-12mm}\left( \frac{1 - \delta_{K + 1}}{\sqrt{K - k}} - \delta_{K + 1} \right) \sqrt{\frac{K - k}{K(1 + \delta_{K + 1})}} \cdot \sqrt{\text{MAR} \cdot \text{SNR}}} \nonumber \\
      &&> \left(1 + \frac{1}{\sqrt{K - k}} \right) \sqrt{1 + \delta_{K + 1}}.  \label{eq:1100012122}
  \end{eqnarray}
  That is,
\begin{equation}
\sqrt{\text{SNR}} > \frac{(1 + \delta_{K + 1}) (\sqrt{K - k} + 1) \sqrt K}{(1 - (\sqrt{K - k} + 1) \delta_{K + 1}) \sqrt{K - k} \cdot \sqrt{\text{MAR}}}. \label{eq:firstite2}
\end{equation}

Furthermore, observe that
  \begin{eqnarray}
\lefteqn{\frac{(1 + \delta_{K + 1}) (\sqrt{K - k} + 1) \sqrt K}{( 1 - (\sqrt{K - k} + 1) \delta_{K + 1}) \sqrt{K - k} \cdot \sqrt{\text{MAR}}}} \nonumber\\
&=& \frac{ \sqrt K (1 + \delta_{K + 1})}{( 1 - (\sqrt{K - k} + 1) \delta_{K + 1})  \cdot \sqrt{\text{MAR}}} \nonumber \\
&&\times \left(1 + \frac{1}{\sqrt{K - k}}\right) \nonumber\\
&\overset{(a)}{<}&
%
%\frac{(1 + \delta_{K + 1})  \sqrt K}{\left( 1 - (\sqrt{K} + 1) \delta_{K + 1} \right)  \cdot \sqrt{\text{MAR}}}  \nonumber \\
%&&\times  \left(1 + \frac{1}{\sqrt{K - k}}\right) \nonumber\\
%&\leq&
\frac{2 \sqrt K (1 + \delta_{K + 1}) }{( 1 - (\sqrt{K} + 1) \delta_{K + 1} )  \cdot \sqrt{\text{MAR}}}, \label{eq:yig2}
  \end{eqnarray}
  where (a) is from the assumption that $1 \leq k < K$ and hence $1 \leq \sqrt{K - k} < \sqrt K$. 
  \vspace{2mm}

Hence, using \eqref{eq:firstite2} and \eqref{eq:yig2}, we can show that $u > v$ is also ensured by
    \begin{equation}
\sqrt{\text{SNR}} > \frac{2 \sqrt K (1 + \delta_{K + 1})  }{( 1 - (\sqrt{K} + 1) \delta_{K + 1} )  \cdot \sqrt{\text{MAR}}},  \label{eq:jjjjffffabb}
  \end{equation}
Therefore, under \eqref{eq:jjjjffffabb}, OMP makes a success at the $(k+1)$-th iteration.
\end{itemize}
\vspace{2mm}

So far, we have obtained condition \eqref{eq:firstite} for the success of the first iteration and condition \eqref{eq:jjjjffffabb} for the success of the general iteration.
We now combing the two conditions to obtain an overall condition that ensures the selection of all support indices with the OMP algorithm.
\vspace{2mm}

Clearly the overall condition will be determined by the more restrictive one between conditions \eqref{eq:firstite} and \eqref{eq:jjjjffffabb}. Thus we compare the right-hand-side of \eqref{eq:firstite} and \eqref{eq:jjjjffffabb}.
Since
\begin{equation}
\frac{2 (1 + \delta_{K + 1})  \sqrt K}{( 1 - (\sqrt{K} + 1) \delta_{K + 1} )  \cdot \sqrt{\text{MAR}}} \geq \frac{(1 + \delta_{K + 1}) (\sqrt K + 1)}{1 - (\sqrt K + 1) \delta_{K + 1}}, \nonumber
\end{equation}
condition \eqref{eq:jjjjffffabb} is more restrictive than \eqref{eq:firstite} and hence becomes the overall condition for the OMP algorithm. The proof is thus complete.

\subsection{Recovery Distortion in $\ell_2$-norm} \label{sec:dist}
When all support indices of $\mathbf{x}$ have been recovered with OMP (i.e., $\mathcal{T}^K = \mathcal{T}$), we have
\begin{eqnarray}
\|\mathbf{x} - \mathbf{x}^K\|_2
& \overset{(a)}{=} & \|\mathbf{x}_{\mathcal{T}^K} - \mathbf{\Phi}^\dag_{\mathcal{T}^K} \mathbf{y}\|_2 \nonumber \\
& {=} & \|\mathbf{x}_{\mathcal{T}^K} - \mathbf{\Phi}^\dag_{\mathcal{T}^K} (\mathbf{\Phi}_{\mathcal{T}} \mathbf{x}_{\mathcal{T}} + \mathbf{v})\|_2 \nonumber \\
& \overset{(b)}{=} & \|\mathbf{x}_{\mathcal{T}^K} - \mathbf{\Phi}^\dag_{\mathcal{T}^K} (\mathbf{\Phi}_{\mathcal{T}^K} \mathbf{x}_{\mathcal{T}^K} + \mathbf{v})\|_2 \nonumber \\
& {=} & \|\mathbf{\Phi}^\dag_{\mathcal{T}^K} \mathbf{v}\|_2 \nonumber \\
& \overset{(c)}{\leq} & \frac{\|\mathbf{\Phi}_{\mathcal{T}^K} \mathbf{\Phi}^\dag_{\mathcal{T}^K} \mathbf{v}\|_2}{\sqrt{1 - \delta_{|\mathcal{T}^K|}}}  \nonumber \\
& = & \frac{\|\mathbf{P}_{\mathcal{T}^K}  \mathbf{v}\|_2}{\sqrt{1 - \delta_{K}}}  \nonumber \\
& \leq & \frac{\|\mathbf{v}\|_2}{\sqrt{1 - \delta_{K}}}, \label{eq:sgag}
\end{eqnarray}
where (a) is because $\mathcal{T}^K = \mathcal{T}$ and $$\mathbf{x}^{K} = \underset{\mathbf{u}:\textit{supp}(\mathbf{u}) = \mathcal{T}^K}{\arg \min} \|\mathbf{y}-\mathbf{\Phi} \mathbf{u}\|_2,$$ (b) is due to  $\mathcal{T}^K = \mathcal{T}$, and (c) is from the RIP. 

One can interpret from \eqref{eq:sgag} that the upper bound of the $\ell_2$-norm recovery distortion with OMP is just proportional with the noise energy, which outperforms the result in~\cite{wang2012Generalized} that suggested a recovery distortion upper bounded by $\mathcal{O}(\sqrt K) \|\mathbf{v}\|_2$.

\section{Approximate Support Recovery via OMP in the Presence of Noise} \label{sec:VI}
In the last section we have shown that the exact support recovery with OMP requires the SNR to scale linearly with the sparsity lever $K$. For high-dimensional setting, this would require the SNR to be unbounded. However, in practical applications, we are often facing with the situation where the SNR is bounded from above.
A particularly interesting case might be the case where the SNR is an absolute constant independent of the problem size. In this case, of course, the necessary condition (in Section~\ref{sec:III}) is not fulfilled so that the exact support recovery of all sparse signals with OMP is impossible. However, we will show that recovery with an arbitrarily small but constant fraction of errors is possible.

The following theorem demonstrates that for properly chosen isometry constants, there exists an absolute constant SNR, under which OMP can approximately recover the support of any $K$-sparse signal with a small constant fraction of errors.  
\vspace{2mm}

\begin{theorem} \label{thm:4}
Let $\kappa := \max_{i,j \in \mathcal{T}} \frac{|x_{i}|}{|x_{j}|}$. Then if $\text{SNR} \geq \kappa^2 \delta_{2K}^{-3/2}$, OMP recovers the support of $K$-sparse signal $\mathbf{x}$ from its noisy measurements $\mathbf{y} = \mathbf{\Phi x} + \mathbf{v}$ with error rate
\begin{equation}
\rho_{\text{error}} \leq C \kappa^2 \delta_{2K}^{1/2} \nonumber
\end{equation}
where $C$ is a constant. 
\end{theorem}
\vspace{2mm}

\begin{remark}
\
\begin{enumerate}[i)]
\item
It is intuitively easy to see that the bound of error rate in Theorem \ref{thm:4} is reasonable because in the special case of orthonormal matrix $\mathbf{\Phi}$ (i.e., $\delta_{2K} = 0$), the result in Theorem \ref{thm:4} suggests that if $\text{SNR} = \infty$, then the error rate $\rho = 0$, which matches with the trivial fact that when there is no noise and $\mathbf{\Phi}$ is an orthonormal matrix, OMP can identify a correct index at each iteration and will accurately recover the whole support of signal $\mathbf{x}$ in exact $K$ iterations. \vspace{2mm}

\item
An interesting point we would like to mention is that our result for the approximate support recovery with OMP only requires the isometry constant $\delta_{2K}$ to be an absolute constant, which essentially imposes a mild constraint on the measurement matrix $\mathbf{\Phi}$. For example, for random Gaussian measurement matrices, it can be satisfied with 
\begin{equation}
 m \geq c K \log \frac{n}{K}
 \end{equation} 
for some constant $c$~\cite{candes2005decoding,baraniuk2008simple}. In CS, this implies that OMP can essentially perform the approximate support recovery of sparse signals with optimal number of random measurements up to a constant. 
\end{enumerate}
\end{remark}
\vspace{2mm}

It is well known that recovering sparse signals with nonzero elements of same magnitude is a particularly challenging case for the OMP algorithm~\cite{dai2009subspace,zhang2011sparse}. The following corollary provides the result of OMP on the approximate support recovery for this type of input signals.
\vspace{2mm}

\begin{corollary} \label{cor:1}
Consider $K$-sparse signals $\mathbf{x}$ with nonzero elements of equal magnitude. Then if $\text{SNR} \geq \delta_{2K}^{-3/2}$, OMP can recover the support of $\mathbf{x}$ from its noisy measurements $\mathbf{y} = \mathbf{\Phi x} + \mathbf{v}$ with error rate
\begin{equation}
\rho_{\text{error}} \leq C \delta_{2K}^{1/2} \nonumber
\end{equation}
where $C$ is a constant.
\end{corollary}
\vspace{2mm}

Our analysis is inspired by the recent work of Livshitz and Temlyakov~\cite{livshitz2014sparse} and will also rely on some proof techniques in~\cite{wang2012how}. The main idea behind our analysis is that most of support indices can essentially be你 identified in $K$ iterations of OMP and the number of false alarms is small. 
%
%, in which it is shown that in the noise-free case and when $K$ is sufficiently large, exact recovery with high probability of random $K$-sparse signals can be achieved within $(1 + \epsilon)K$ iterations of the OMP algorithm. Also, our proof relies on some proof techniques developed in~\cite{wang2012how}. 
%
% 
The proof of Theorem~\ref{thm:4} follows along a similar line as the proofs in~\cite{livshitz2014sparse}, but with three important distinctions. Firstly, the main goals of proofs are different. While our analysis is based on the approximate support recovery of $K$-sparse signals with $K$ iterations of OMP, the analysis in~\cite{livshitz2014sparse} concerned the exact support recovery with OMP in more than $K$ iterations. Secondly, compared to the result in~\cite{livshitz2014sparse}, our result is more general in that it applies to input signals with arbitrary sparsity level $K$ and with nonzero elements of arbitrary magnitudes. Note that the analysis of~\cite{livshitz2014sparse} assumed that $K \geq \delta_{2K}^{-1/2}$, which essentially applies to the situation where the sparsity $K$  of the input signal is nontrivial. In addition,~\cite{livshitz2014sparse} considered only the recovery of signals with magnitudes of nonzero elements upper bounded by one. Thirdly, and most importantly, we consider the scenario where the measurement noise is present and build conditions based on the SNR. Whereas, the analyses in~\cite{livshitz2014sparse} focused only on the situation without noise.

\subsection{Proof of Theorem~\ref{thm:4}}

Before we proceed to the details of the proof, we introduce some useful notations and definitions. For notational simplicity, let $\delta:= \delta_{2K}$. At the $k$-th iteration ($0 \leq k \leq K$), let $\Gamma^k: = \mathcal{T} \backslash \mathcal{T}^k$ denote the set of missed detection of support indices. For given constant $\tau \in (0, 1]$, let $\Gamma_\tau^k$ denote the subset of $\Gamma^k$ corresponding to the $\lceil \tau K \rceil$ largest elements (in magnitude) of $\mathbf{x}_{\Gamma^k}$. Also, let $x_\tau^k$ denote the $\lceil \tau K \rceil$-th largest element (in magnitude) in $\mathbf{x}_{\Gamma^k}$. Following the idea in~\cite{livshitz2014sparse}, we will fix $\tau = \delta^{1/2}$. If $\lceil \tau K \rceil > |\Gamma^k|$, then set $\Gamma^k_\tau = \Gamma^k$ and $x_\tau^k = 0$.
Since OMP totally runs $K$ iterations before stopping, the error rate of the support recovery can be given by
\begin{equation}
\rho_{\text{error}} := \frac{|\mathcal{T}^K \backslash \mathcal{T}|}{|\mathcal{T}|}.
\end{equation}

The proof of Theorem~\ref{thm:4} consists of two parts. In the first part, we will provide a lower bound on the reduction of residual energy at each iteration of OMP ({Proposition~\ref{lem:1}}). In the second part, by means of the lower bound obtained in Proposition~\ref{lem:1}, we will estimate the remaining energy in the residual vector $\mathbf{r}^K$. The estimate of the energy of $\mathbf{r}^K$ will then allow us to derive an upper bound on the number of missed detections (i.e., $|\Gamma^K|$). Since the OMP algorithm totally chooses $K$ indices, it is easy to see that the number missed detections (after $K$ iterations) is equal to the number of false alarms. i.e., 
\begin{equation}
|\Gamma^K| = |\mathcal{T}^K \backslash \mathcal{T}|,
\end{equation} 
which implies that 
\begin{equation}
\rho_{\text{error}} = \frac{|\Gamma^K|}{K}.
\end{equation} 
Therefore, from the upper bound of $|\Gamma^K|$ we can directly obtain an upper bound of the error rate for the support recovery with OMP.
\vspace{2mm}

\begin{proposition} \label{lem:1}
For any $0 \leq k \leq K - \lceil \delta^{1/2} K \rceil$, the residual of OMP satisfies
\begin{equation}
 \| \mathbf{r}^k \|_2^2 - \| \mathbf{r}^{k + 1} \|_2^2 \geq (1 - 7 \delta^{1/2} ) \left({x} ^k_{\delta^{1/2}}\right)^2. \nonumber
\end{equation}
\end{proposition}
\vspace{2mm}

\begin{proof}
We shall prove Proposition~\ref{lem:1} in two steps. First, we show that the residual
power difference of OMP satisfies (see Appendix~\ref{app:fir})\footnote{This proof is essentially identical to a result in \cite{wang2012how}. Since it will play a key role in the proof of \eqref{eq:sec0}, we include the proof for completeness.}
\begin{equation} \label{eq:fir}
  \|\mathbf{r}^k\|_2^2 - \|\mathbf{r}^{k + 1}\|_2^2 \geq \frac{\|\mathbf{\Phi}' \mathbf{r}^k\|_\infty^2}{1 + \delta_{1}}.
\end{equation}

In the second step, we show that (see Appendix \ref{app:sec0})
  \begin{equation} \label{eq:sec0}
\|\mathbf{\Phi}' \mathbf{r}^k \|_\infty^2 \geq (1 - 7 \tau ) \left({x} ^k_{\tau} \right)^2.
  \end{equation}
Using \eqref{eq:fir} and \eqref{eq:sec0},
we have
  \begin{eqnarray} \label{eq:lemr1}
 \| \mathbf{r}^k \|_2^2 - \| \mathbf{r}^{k + 1} \|_2^2 &\geq&  \left(\frac{1 - 6 \tau}{1 + \delta_1} \right) ({x} ^k_{\tau})^2 \nonumber \\
 & \overset{(a)}{\geq} & \left(\frac{1 - 6 \tau}{1 + \delta} \right) ({x} ^k_{\tau})^2 \nonumber \\
 &{=} & \left(\frac{1 - 6 \tau}{1 + \tau^2} \right) ({x} ^k_{\tau})^2 \nonumber \\
 & = & \left(1 - 7 \tau + \frac{\tau - \tau^2 + 7 \tau^3}{1 + \tau^2} \right) ({x} ^k_{\tau})^2 \nonumber \\
 &\overset{(b)}{\geq}& \left(1 - 7 \tau \right) ({x} ^k_{\tau})^2 \nonumber \\
 &=& (1 - 7 \delta^{1/2} ) \left({x} ^k_{\delta^{1/2}} \right)^2,
\end{eqnarray}
where (a) is because $\delta =\delta_{2K} \geq \delta_1$ and (b) is from $\tau = \delta^{1/2} \in [0,1)$,
which establishes the proposition.
\end{proof}
\vspace{3mm}

In Proposition \ref{lem:1}, we have shown that each iteration of OMP makes non-trivial progress by providing the lower bound on the reduction of residual energy at each iteration. Next, using the bound obtained in Proposition \ref{lem:1}, we will derive an upper bound on the number of missed detections in the support recovery of OMP.
Without loss of generality we assume that $\mathcal{T} = \{1, \cdots, K\}$ and that the elements of $\{x_i\}_{i = 1}^K$ are in a descending order of their magnitudes.  Then from the definition of $x_\tau^k$ we have that for any $k \geq 0$, $k + \lceil \tau K \rceil \leq K$,
\begin{equation}
|x_\tau^k| \geq |x_{k + \lceil \tau K \rceil}|. \label{eq:22good}
\end{equation}
By applying Proposition~\ref{lem:1}, we have
\begin{eqnarray}
\|\mathbf{r}^K\|_2^2 & = & \|\mathbf{r}^0\|_2^2 - \sum_{k = 0}^{K - 1} (\|\mathbf{r}^k\|_2^2 - \|\mathbf{r}^{k + 1}\|_2^2 ) \nonumber \\
& \overset{(a)}{\leq} &  \|\mathbf{y}\|_2^2 - \sum_{k = 1}^{K - \lceil \tau K \rceil} (\|\mathbf{r}^k\|_2^2 - \|\mathbf{r}^{k + 1}\|_2^2 ) \nonumber \\
& \leq & \|\mathbf{y}\|_2^2 - \sum_{k = 1}^{K - \lceil \tau K \rceil} \left(1 - 7 \tau \right) ({x} ^k_{\tau})^2  \nonumber \\
& \overset{(b)}{\leq} & \|\mathbf{y}\|_2^2 - \sum_{k = 1}^{K - \lceil \tau K \rceil} \left(1 - 7 \tau \right) (x_{k + \lceil \tau K \rceil})^2 \nonumber \\
& = & \|\mathbf{y}\|_2^2 - \sum_{k = \lceil \tau K \rceil + 1}^{K} \left(1 - 7 \tau \right) (x_i)^2 \label{eq:23eq}
\end{eqnarray}
where (a) uses the facts that $\lceil \tau K \rceil \geq 1$ and that the energy of residual of the OMP algorithm is always non-increasing with the number of iterations (i.e., $\|\mathbf{r}^k\|_2^2 \geq \|\mathbf{r}^{k + 1}\|_2^2$, $k \geq 0$), and (b) is from \eqref{eq:22good}.

Note that
\begin{eqnarray}
\|\mathbf{y}\|_2^2
& = & \|\mathbf{\Phi x + v}\|_2^2  \nonumber \\
& \overset{(a)}{\leq} & (1 + \tau) \|\mathbf{\Phi x}\|_2^2 + \left(1 + {1}/{\tau} \right) \|\mathbf{v}\|_2^2 \nonumber \\
%& \overset{(b)}{\leq} & (1 + \tau) \|\mathbf{\Phi x}\|_2^2 + \tau^2 (1 + 3 \tau) K (x_{\min})^2 \nonumber \\
& \overset{(b)}{\leq} & (1 + \tau) (1 - \delta) \|\mathbf{x}\|_2^2 + \left(1 + {1}/{\tau} \right) \|\mathbf{v}\|_2^2 \nonumber \\
& \overset{(c)}{\leq} & (1 + 3\tau) \|\mathbf{x}\|_2^2 + \left(1 + {1}/{\tau} \right) \|\mathbf{v}\|_2^2 \nonumber \\
& = & (1 + 3\tau) \sum_{i = 1}^K (x_i)^2 + \left(1 + {1}/{\tau} \right) \|\mathbf{v}\|_2^2,
\end{eqnarray}
where (a) is from the fact that $$\|\mathbf{u} + \mathbf{v}\|_2^2 \leq (1 + \tau) \|\mathbf{u}\|_2^2 +\left(1 + {1}/{\tau}\right) \|\mathbf{v}\|_2^2$$
with $\mathbf{u} = \mathbf{\Phi x}$, (b) is due to the RIP,
and (c) is from that $$(1 + \tau)(1 + \delta) \leq (1 + \tau)^2 \leq 1 + 3 \tau.$$

Hence, we can rewrite \eqref{eq:23eq} as
\begin{eqnarray}
\|\mathbf{r}^K\|_2^2 & \leq & \|\mathbf{y}\|_2^2 - \sum_{k = \lceil \tau K \rceil + 1}^{K} \left(1 - 7 \tau \right) (x_i)^2 \nonumber \\
& {\leq} & (1 + 3\tau) \sum_{i = 1}^K (x_i)^2 - \sum_{i = \lceil \tau K \rceil + 1}^{K} \left(1 - 7 \tau \right) (x_i)^2   \nonumber \\
 && + \left(1 + {1}/{\tau} \right) \|\mathbf{v}\|_2^2 \nonumber \\
& \leq & 10 \tau \sum_{i = 1}^K (x_i)^2 + \sum_{i = 1}^{\lceil \tau K \rceil}   (x_i)^2 + \left(1 + {1}/{\tau} \right) \|\mathbf{v}\|_2^2 \nonumber \\
& \leq & (10 \tau K + {\lceil \tau K \rceil})   (x_{\max})^2     + \left(1 + {1}/{\tau} \right) \|\mathbf{v}\|_2^2 \nonumber \\
& \leq & 11 {\lceil \tau K \rceil}   (x_{\max})^2  + \left(1 + {1}/{\tau} \right) \|\mathbf{v}\|_2^2. \label{eq:25me} 
\end{eqnarray}

On the other hand,
\begin{eqnarray}
\|\mathbf{r}^K\|_2^2 & = & \|\mathbf{\Phi} (\mathbf{x} - \mathbf{x}^K) + \mathbf{v}\|_2^2 \nonumber \\
& \overset{(a)}{\geq} & (1 - \tau) \|\mathbf{\Phi} (\mathbf{x} - \mathbf{x}^K)\|_2^2 - \left( {1}/{\tau} - 1 \right) \|\mathbf{v}\|_2^2 \nonumber \\
& \overset{(b)}{\geq} & (1 - \tau) (1 - \delta)\|\mathbf{x} - \mathbf{x}^K\|_2^2 - \left( {1}/{\tau} - 1 \right) \|\mathbf{v}\|_2^2 \nonumber \\
& \overset{(c)}{\geq} & (1 - 2\tau) \|\mathbf{x} - \mathbf{x}^K\|_2^2 - \left( {1}/{\tau} - 1 \right) \|\mathbf{v}\|_2^2 \nonumber \\
& {\geq} & (1 - 2\tau) \|(\mathbf{x} - \mathbf{x}^K)_{\Gamma^K}\|_2^2 - \left( {1}/{\tau} - 1 \right) \|\mathbf{v}\|_2^2 \nonumber \\
& \overset{(d)}{\geq} & (1 - 2\tau) \|\mathbf{x}_{\Gamma^K}\|_2^2 - \left( {1}/{\tau} - 1 \right) \|\mathbf{v}\|_2^2, \label{eq:26me}
\end{eqnarray} 
where (a) uses the fact that $$\|\mathbf{u} + \mathbf{v}\|_2^2 \geq (1 - \tau) \|\mathbf{u}\|_2^2 - \left( {1}/{\tau} - 1\right) \|\mathbf{v}\|_2^2$$
with $\mathbf{u} = \mathbf{\Phi} (\mathbf{x} - \mathbf{x}^K)$, (b) follows from the RIP, (c) is because $$(1 - \tau)(1 - \delta) \geq (1 - \tau)^2 = 1 - 2 \tau + \tau^2 \geq 1 - 2 \tau,$$ and (d) is due to the fact that $\mathbf{x}^K$ is supported on $\mathcal{T}^K$ and hence $\mathbf{x}^K_{\Gamma^K} = \mathbf{x}^K_{\mathcal{T} \backslash \mathcal{T}^K} = \mathbf{0}$.
\vspace{2mm}

Using \eqref{eq:25me} and \eqref{eq:26me}, we have
\begin{eqnarray}
\lefteqn{\|\mathbf{x}_{\Gamma^K}\|_2^2} \nonumber \\
&\leq& \frac{1}{1 - 2\tau} \left(11 {\lceil \tau K \rceil}   (x_{\max})^2  + \frac{2}{\tau} \|\mathbf{v}\|_2^2\right) \nonumber \\
&\overset{(a)}{\leq}& \frac{1}{1 - 2\tau} \left(11 {\lceil \tau K \rceil}   (x_{\max})^2 + 2 \tau^2 (1 + \tau^2) K (x_{\min})^2 \right) \nonumber \\
&{\leq}& \frac{1}{1 - 2\tau} \left(11 {\lceil \tau K \rceil} \kappa^2  + 2 \tau^2 (1 + \tau^2) K  \right) (x_{\min})^2 \nonumber \\
&{\leq}& \left(\frac{11  \kappa^2  + 2 \tau (1 + \tau^2)}{1 - 2\tau} \right) {\lceil \tau K \rceil} (x_{\min})^2 \nonumber \\
&\leq& C \kappa^2 \tau K (x_{\min})^2 \nonumber \\
&=& C \kappa^2 \delta^{1/2} K (x_{\min})^2,
\end{eqnarray}
where (a) is due to the facts 1) that $$\text{SNR} = \frac{\|\mathbf{\Phi x}\|_2^2}{\|\mathbf{v}\|_2^2} \geq \kappa^2 \delta_{2K}^{-3/2} =  \frac{\kappa^2}{\tau^3},$$ and 2) that
\begin{eqnarray}
\|\mathbf{\Phi x}\|_2^2  & \overset{\text{RIP}}{\leq} & (1 + \delta) {\|\mathbf{x}\|_2^2} \nonumber \\
& {\leq} & (1 + \delta) K (x_{\max})^2 \nonumber \\
& {\leq} & (1 + \delta) K (\kappa x_{\min})^2, \nonumber
\end{eqnarray}
and hence
\begin{eqnarray}
\frac{2}{\tau} \|\mathbf{v}\|_2^2
&\leq& 2 \tau^2 (1 + \delta) K (x_{\min})^2 \nonumber \\
&=& 2 \tau^2 (1 + \tau^2) K (x_{\min})^2. \nonumber
\end{eqnarray}

Finally, by noting that \begin{equation}
\|\mathbf{x}_{\Gamma^K}\|_2^2 \geq |\Gamma^K| (x_{\min})^2, \nonumber
\end{equation}
we have that the number of missed detections satisfies
\begin{equation}
|\Gamma^K| \leq C \kappa^2 \delta^{1/2} K.
\end{equation}
Recall that the number missed detections (after $K$ iterations) is essentially equal to the number of false alarms (i.e., $|\Gamma^K| = |\mathcal{T}^K \backslash \mathcal{T}|$).
Therefore, the error rate of support recovery with OMP satisfies
\begin{equation}
\rho_{\text{error}} = \frac{|\mathcal{T}^K \backslash \mathcal{T}|}{|\mathcal{T}|} = \frac{|\Gamma^K|}{K} \leq C \kappa^2 \delta^{1/2}.
\end{equation}
The proof is now complete.

\section{Conclusion} \label{sec:V}

In this paper, we have studied the performance of OMP for the support recovery of sparse signals under noise. In the first part of our analysis, we have shown that in order for the OMP algorithm to accurately recover the support of any $K$-sparse signal, the SNR must be at least proportional to the sparsity level $K$ of the signal. For high-dimensional setting, our result indicates that the exact support recovery with OMP is not possible under finite SNR. 

In the second part of our analysis, we have considered a practical scenario where the SNR is an absolute constant independent of the sparsity $K$. While the exact support recovery with OMP is not possible for this scenario, our analysis has shown that recovery with an arbitrarily small but constant fraction of errors is possible. For high-dimensional setting, our result offers an affirmative answer to the open question of whether OMP can perform the approximate support recovery of sparse signals with bounded SNR~\cite{fletcher2012orthogonal}. We would like to point out a technical limitation in this result. Unlike existing results for the exact support recovery that depend on the minimum magnitude of nonzero elements in the signal, our result for the approximate support recovery exhibits the dependence on the minimum as well as the maximum magnitudes (more precisely, the ratio $\kappa$). Deriving a similar result but without the dependence on the maximum magnitude would require a more refined analysis and our future work will be directed towards this avenue.

\appendices

\section{Proof of Theorem~\ref{thm:3}} \label{app:thm3}

\begin{proof}
To prove the necessity of the lower bound of SNR in \eqref{eq:jjjjffffa}, it suffices to show that OMP may fail to recover the support of some sparse signal $\mathbf{x}$ when
\begin{equation}
\sqrt{\text{SNR}} \leq \frac{\sqrt{K}(1 + \delta_{K + 1})}{( 1 - \sqrt{K} \delta_{K + 1} ) \cdot \sqrt{\text{MAR}}}.  \label{eq:jjjjffffaf}
  \end{equation}
In the following, we will show that there exists a set of  $\mathbf{\Phi}$, $\mathbf{x}$, and $\mathbf{v}$, for which \eqref{eq:jjjjffffaf} is satisfied but OMP fails to recover the support of $\mathbf{x}$. 
\vspace{2mm}

Consider an identity matrix $\mathbf{\Phi}^{m \times m}$, a $K$-sparse signal $\mathbf{x} \in \mathcal{R}^m$ with all nonzero elements equal to one, and an $1$-sparse noise vector $\mathbf{v} \in \mathcal{R}^m$ as follows,
    \begin{equation}
  \mathbf{\Phi} = \left[ \begin{array}{cccc}
    1 &      &        &   \\
      & 1    &        &   \\
      &      & \ddots &   \\
      &      &        & 1
  \end{array} \right],   \mathbf{x} = \left[ \begin{array}{c}
    1\\
    \vdots\\
    1\\
    0\\
    \vdots \\
    0
  \end{array} \right], ~~\text{and}~~
  \mathbf{v} = \left[ \begin{array}{c}
    0\\
    \vdots \\
    0 \\
    1
  \end{array} \right]. \nonumber
\nonumber
\end{equation}
Then the measurements are given by
\begin{equation}
  \mathbf{y} = \left[ \begin{array}{c}
    1\\
    \vdots\\
    1\\
    0\\
    \vdots \\
    0\\
    1
  \end{array} \right]. \nonumber
\end{equation}
In this case, we have $\delta_{K + 1} = 0$,
\begin{equation}
\text{SNR} = \frac{\| \mathbf{\Phi} \mathbf{x} \|_2^2}{\| \mathbf{v} \|_2^2} = K ~~\text{and}~~ {\text{MAR}} = 1.
\nonumber
\end{equation}
It is easily verified that condition \eqref{eq:jjjjffffaf} is satisfied; however, OMP fails to recover the support of $\mathbf{x}$. Specifically, OMP is not guaranteed to make a correct selection at the first iteration.
\end{proof}

\section{Proof of Proposition \ref{prop:upperbound15}}\label{app:upperbound1}

\begin{proof}
We first give a proof of {\eqref{eq:small5}}. Since $u$ is the largest value in $\{ |\langle \phi_{i},
    \mathbf{r}^{k} \rangle| \}_{i \in \mathcal{T} \backslash \mathcal{T}^k}$,
    \begin{eqnarray}
u & = & \max_{i \in \mathcal{T} \backslash \mathcal{T}^k} |\langle \phi_{i}, \mathbf{r}^{k} \rangle| \nonumber\\
      & {\geq} & \frac{1}{\sqrt{|\mathcal{T} \backslash \mathcal{T}^k|}}  \sqrt{\sum_{i \in \mathcal{T} \backslash \mathcal{T}^k}
      \langle \phi_{i}, \mathbf{r}^{k} \rangle^2} \nonumber\\
      & \overset{(a)}{=} & \frac{1}{\sqrt{K - k}} \| \mathbf{\Phi}'_{\mathcal{T} \backslash \mathcal{T}^k}
      \mathbf{r}^{k} \|_{2} \nonumber\\
      & = & \frac{1}{\sqrt{K - k}} \| \mathbf{\Phi}'_{\mathcal{T} \backslash \mathcal{T}^k}
      \mathbf{P}^\bot_{\mathcal{T}^k} (\mathbf{\Phi x} + \mathbf{v}) \|_{2} \nonumber\\
      & \overset{(b)}{\geq} & \frac{1}{\sqrt{K - k}}
      \left( \| \mathbf{\Phi}'_{\mathcal{T} \backslash \mathcal{T}^k} \mathbf{P}^\bot_{\mathcal{T}^k}  \mathbf{\Phi} \mathbf{x} \|_2 - \| \mathbf{\Phi}'_{\mathcal{T} \backslash \mathcal{T}^k} \mathbf{P}^\bot_{\mathcal{T}^{k}} \mathbf{v} \|_{2} \right), \nonumber\\ \label{eq:ggeeewwwq}
    \end{eqnarray}
where (a) is from \eqref{eq:200002} and (b) is due to the triangle inequality.

Observe that
    \begin{eqnarray}
 \lefteqn{ \| \mathbf{\Phi}'_{\mathcal{T} \backslash \mathcal{T}^k} \mathbf{P}^\bot_{\mathcal{T}^k} \mathbf{\Phi} \mathbf{x} \|_2 } \nonumber \\  & \overset{(a)}{=} &
  \| \mathbf{\Phi}'_{\mathcal{T} \backslash \mathcal{T}^k} \mathbf{P}^\bot_{\mathcal{T}^k} \mathbf{\Phi}_{\mathcal{T} \backslash \mathcal{T}^k} \mathbf{x}_{\mathcal{T} \backslash \mathcal{T}^k}\|_2 \nonumber \\
  & \overset{(b)}{\geq} &
  \frac{\|  \mathbf{x}'_{\mathcal{T} \backslash \mathcal{T}^k} \mathbf{\Phi}'_{\mathcal{T} \backslash \mathcal{T}^k} \mathbf{P}^\bot_{\mathcal{T}^k} \mathbf{\Phi}_{\mathcal{T} \backslash \mathcal{T}^k} \mathbf{x}_{\mathcal{T} \backslash \mathcal{T}^k}\|_2}{\| \mathbf{x}'_{\mathcal{T} \backslash \mathcal{T}^k}\|_2} \nonumber \\
  & \overset{(c)}{=} &
  \frac{\|  \mathbf{x}'_{\mathcal{T} \backslash \mathcal{T}^k} \mathbf{\Phi}'_{\mathcal{T} \backslash \mathcal{T}^k} (\mathbf{P}^\bot_{\mathcal{T}^k})' \mathbf{P}^\bot_{\mathcal{T}^k} \mathbf{\Phi}_{\mathcal{T} \backslash \mathcal{T}^k} \mathbf{x}_{\mathcal{T} \backslash \mathcal{T}^k}\|_2}{\| \mathbf{x}_{\mathcal{T} \backslash \mathcal{T}^k}\|_2} \nonumber \\
  & = &
  \frac{\|  \mathbf{P}^\bot_{\mathcal{T}^k} \mathbf{\Phi}_{\mathcal{T} \backslash \mathcal{T}^k} \mathbf{x}_{\mathcal{T} \backslash \mathcal{T}^k}\|_2^2}{\| \mathbf{x}_{\mathcal{T} \backslash \mathcal{T}^k}\|_2} \nonumber \\
  & \geq & \frac{\lambda_{\min} \left( (\mathbf{P}^\bot_{\mathcal{T}^k} \mathbf{\Phi}_{\mathcal{T} \backslash \mathcal{T}^k})' \mathbf{P}^\bot_{\mathcal{T}^k} \mathbf{\Phi}_{\mathcal{T} \backslash \mathcal{T}^k} \right)  \| \mathbf{x}_{\mathcal{T} \backslash \mathcal{T}^k} \|_{2}^2 }{\| \mathbf{x}_{\mathcal{T} \backslash \mathcal{T}^k}\|_2} \nonumber \\
  & \overset{(d)}{=} & \lambda_{\min} \left( \mathbf{\Phi}'_{\mathcal{T} \backslash \mathcal{T}^k} \mathbf{P}^\bot_{\mathcal{T}^k} \mathbf{\Phi}_{\mathcal{T} \backslash \mathcal{T}^k} \right)  \| \mathbf{x}_{\mathcal{T} \backslash \mathcal{T}^k} \|_{2} \nonumber \\
  & \overset{(e)}{\geq} & \lambda_{\min} ({ \mathbf{\Phi}'_{\mathcal{T}
      \cup \mathcal{T}^{k}} \mathbf{\Phi}_{\mathcal{T} \cup \mathcal{T}^{k}}}) \| \mathbf{x}_{\mathcal{T} \backslash \mathcal{T}^k} \|_{2} \nonumber\\
  & \overset{(f)}{\geq} & (1 - \delta_{K})\| \mathbf{x} _{\mathcal{T} \backslash \mathcal{T}^k} \|_{2},
      \label{eq:geaig1aa}
    \end{eqnarray}
    where  (a) is because $\mathbf{P}^\bot_{\mathcal{T}^k} \mathbf{\Phi}_{\mathcal{T}^k} = \mathbf{0}$, (b) is from the norm inequality, (c) and (d) use the fact that $ \mathbf{P}^\bot_{\mathcal{T}^k} =  (\mathbf{P}^\bot_{\mathcal{T}^k} )' =  (\mathbf{P}^\bot_{\mathcal{T}^k} )^2$, (e) is from Lemma~\ref{lem:rip6}, and (f) is from the RIP. (Note that $\mathcal{T}^{k} \subset \mathcal{T}$ and so $|\mathcal{T} \cup \mathcal{T}^{k}| = |\mathcal{T}| = K$.)
Also,
\begin{eqnarray}
 \lefteqn{\left\| \mathbf{\Phi}'_{\mathcal{T} \backslash \mathcal{T}^k} \mathbf{P}^\bot_{\mathcal{T}^{k}} \mathbf{v} \right\|_{2}} \nonumber\\
 & \overset{(a)}{=} & \left\| \left( \mathbf{P}^\bot_{\mathcal{T}^{k}} \mathbf{\Phi}_{\mathcal{T} \backslash \mathcal{T}^k} \right)' \mathbf{v} \right\|_{2} \nonumber\\
  & {\leq} & \sqrt{\lambda_{\max} \left( \left( \mathbf{P}^\bot_{\mathcal{T}^{k}} \mathbf{\Phi}_{\mathcal{T} \backslash \mathcal{T}^k} \right)' \mathbf{P}^\bot_{\mathcal{T}^{k}} \mathbf{\Phi}_{\mathcal{T} \backslash \mathcal{T}^k} \right)} \left\| \mathbf{v} \right\|_{2} \nonumber \\
 & \overset{(b)}{=} & \sqrt{\lambda_{\max} \left(  \mathbf{\Phi}'_{\mathcal{T}  \backslash \mathcal{T}^k } \mathbf{P}^\bot_{\mathcal{T}^{k}} \mathbf{\Phi}_{\mathcal{T}  \backslash \mathcal{T}^k } \right)} \left\| \mathbf{v} \right\|_{2} \nonumber \\
 & \overset{(c)}{\leq} & \sqrt{\lambda_{\max} \left( \mathbf{\Phi}'_{\mathcal{T} \cup \mathcal{T}^k  }  \mathbf{\Phi}_{\mathcal{T} \cup \mathcal{T}^k } \right)} \left\| \mathbf{v} \right\|_{2} \nonumber \\
 & {\leq} & \sqrt{1 + \delta_{K}} ~\left\| \mathbf{v} \right\|_{2}  \label{eq:969}
\end{eqnarray}
where (a) and (b) are due to $\mathbf{P}^\bot_{\mathcal{T}^k} =  (\mathbf{P}^\bot_{\mathcal{T}^k} )' =  (\mathbf{P}^\bot_{\mathcal{T}^k} )^2$ and (c) is from Lemma~\ref{lem:rip6}.

Using \eqref{eq:ggeeewwwq}, \eqref{eq:geaig1aa} and \eqref{eq:969}, we have
\begin{equation}
u \geq \frac{1}{\sqrt{K - k}} \left( (1 - \delta_{K}) \| \mathbf{x}_{\mathcal{T} \backslash \mathcal{T}^k} \|_{2} - \sqrt{1 + \delta_{K}} ~\| \mathbf{v} \|_{2} \right).
\end{equation} 
\vspace{2mm}

Next, we proceed to prove \eqref{eq:large5}.
Let $w \in \Omega$ denote the index corresponding to the largest element in $\left\{ |\langle \phi_{i}, \mathbf{r}^{k} \rangle| \right \}_{i \in \Omega \setminus (\mathcal{T} \cup \mathcal{T}^{k})}$. Then by the definition of $v$,
  \begin{eqnarray}
v & = & |\langle \phi_{w}, \mathbf{r}^{k} \rangle| \nonumber\\
  & = & |\langle \phi_{w}, \mathbf{P}^\bot_{\mathcal{T}^{k}} (\mathbf{\Phi x} + \mathbf{v}) \rangle| \nonumber \\
  & \overset{(a)}{\leq} & |\langle \phi_{w}, \mathbf{P}^\bot_{\mathcal{T}^{k}} \mathbf{\Phi x} \rangle| + |\langle \phi_{w}, \mathbf{P}^\bot_{\mathcal{T}^{k}}  \mathbf{v} \rangle| \nonumber \\
  & \overset{(b)}{=} & |\langle \phi_{w},  \mathbf{P}^\bot_{\mathcal{T}^{k}} \mathbf{\Phi}_{\mathcal{T} \setminus \mathcal{T}^{k}} \mathbf{x}_{\mathcal{T} \setminus \mathcal{T}^{k}} \rangle| + |\langle \phi_{w}, \mathbf{P}^\bot_{\mathcal{T}^{k}} \mathbf{v} \rangle|, 
  \label{eq:10009}
  \end{eqnarray}
where (a) is from the triangle inequality and (b) is because $\mathbf{P}^\bot_{\mathcal{T}^{k}} \mathbf{\Phi}_{\mathcal{T}^k} = \mathbf{0}$.
\vspace{2mm}

Since $w \notin \mathcal{T}$ and also noting that $\mathcal{T}^{k} \subset \mathcal{T}$ (by the induction hypothesis), we have $$|\{w\} \cup \mathcal{T}^{k} \cup (\mathcal{T} \setminus \mathcal{T}^{k})| = |\{w\} \cup \mathcal{T}| = K + 1.$$ Using this together with Lemma~\ref{lem:correlationrip}, we have
    \begin{eqnarray}
      \lefteqn{|\langle \phi_{w},  \mathbf{P}^\bot_{\mathcal{T}^{k}} \mathbf{\Phi}_{\mathcal{T} \setminus \mathcal{T}^{k}} \mathbf{x}_{\mathcal{T} \setminus \mathcal{T}^{k}} \rangle|} \nonumber\\
      & = & \left\| \phi'_{w} ( \mathbf{P}^\bot_{\mathcal{T}^{k}} \mathbf{\Phi}_{\mathcal{T} \setminus \mathcal{T}^{k}} ) \mathbf{x}_{\mathcal{T} \setminus \mathcal{T}^{k}} \right\|_{2} \nonumber\\
      &\leq& \delta_{|\{w\} \cup \mathcal{T}^{k} \cup (\mathcal{T} \setminus \mathcal{T}^{k}) |} \| \mathbf{x}_{\mathcal{T} \setminus \mathcal{T}^{k}} \|_{2} \nonumber \\
      &=& \delta_{K + 1} \| \mathbf{x}_{\mathcal{T} \backslash \mathcal{T}^k} \|_{2}  \label{eq:j1}
    \end{eqnarray}
and 
\begin{eqnarray}
|\langle \phi_{w}, \mathbf{P}^\bot_{\mathcal{T}^{k}} \mathbf{v} \rangle| &=&
|\langle (\mathbf{P}^\bot_{\mathcal{T}^{k}})' \phi_{w}, \mathbf{v} \rangle| \nonumber\\
&\overset{(a)}{=}&
|\langle \mathbf{P}^\bot_{\mathcal{T}^{k}} \phi_{w}, \mathbf{v} \rangle| \nonumber\\
&\overset{(b)}{\leq}& \| \mathbf{P}^\bot_{\mathcal{T}^{k}} \phi_{w} \|_2 ~\|\mathbf{v} \|_2 \nonumber \\
&{\leq}& \| \phi_{w} \|_2 ~\|\mathbf{v} \|_2 \nonumber \\
&\overset{(c)}{\leq}& \sqrt{1 + \delta_{1}} ~\|\mathbf{v} \|_2, \label{eq:ghg2553}
\end{eqnarray}
where (a) is due to the fact that $\mathbf{P}^\bot_{\mathcal{T}^k} =  (\mathbf{P}^\bot_{\mathcal{T}^k} )'$, (b) is from Cauchy-Schwarz inequality, and (c) is from the RIP.
\vspace{2mm}

Finally, combining \eqref{eq:10009}, (\ref{eq:j1}), and \eqref{eq:ghg2553} we have
\begin{equation}
    v \leq \delta_{K + 1} \| \mathbf{x}_{\mathcal{T} \backslash \mathcal{T}^k} \|_{2} + \sqrt{1 + \delta_{1}} ~\|\mathbf{v}\|_2,
 \label{eq:left5}
 \end{equation}
 which completes the proof.
\end{proof}

\section{Proof of \eqref{eq:fir}} \label{app:fir}

\begin{proof}
First, from the definition of OMP (see Table~\ref{tab:omp}), we have that for any integer $k \geq 0$,
\begin{eqnarray}
 \mathbf{r}^k - \mathbf{r}^{k + 1} &=& (\mathbf{y} - \mathbf{\Phi} \mathbf{x}^{k}) - (\mathbf{y} - \mathbf{\Phi} \mathbf{x}^{k + 1}) \nonumber \\
&\overset{(a)}{=}& \mathbf{P}_{\mathcal{T}^{k}} \mathbf{y} -
\mathbf{P}_{\mathcal{T}^{k + 1}} \mathbf{y} \nonumber \\
&\overset{(b)}{=}& (\mathbf{P}_{\mathcal{T}^{l + 1}} - \mathbf{P}_{\mathcal{T}^{l + 1}}
\mathbf{P}_{\mathcal{T}^{l}} )\mathbf{y} \nonumber \\ &=& \mathbf{P}_{\mathcal{T}^{l +
1}}  \mathbf{P}^\bot_{\mathcal{T}^{l}}
\mathbf{y} \nonumber \\ 
&=& \mathbf{P}_{\mathcal{T}^{l + 1}} \mathbf{r}^k,
\end{eqnarray}
where (a) is from that $$\mathbf{x}^{k} = \underset{\mathbf{u}:\textit{supp}(\mathbf{u}) = \mathcal{T}^k}{\arg \min} \|\mathbf{y}-\mathbf{\Phi} \mathbf{u}\|_2$$ and hence $\mathbf{\Phi} \mathbf{x}^{k} =  \mathbf{\Phi}_{\mathcal{T}^k} \mathbf{\Phi}^\dag_{\mathcal{T}^k} \mathbf{y} = \mathbf{P}_{\mathcal{T}^{l}} \mathbf{y}$, (b) is because $\textit{span}(\mathbf{\Phi}_{\mathcal{T}^{k}})
\subseteq \textit{span}(\mathbf{\Phi}_{\mathcal{T}^{k + 1}})$ so that
$\mathbf{P}_{\mathcal{T}^k} \mathbf{y} = \mathbf{P}_{\mathcal{T}^{k + 1}}
(\mathbf{P}_{\mathcal{T}^{k}} \mathbf{y})$.\\

Since $\textit{span}(\mathbf{\Phi}_{\mathcal{T}^{k + 1}}) \supseteq
\textit{span}({\phi}_{t^{k + 1}})$, we have
\begin{eqnarray} \label{eq:residual_A}
  \|\mathbf{r}^k - \mathbf{r}^{k + 1}\|_2^2 = \|\mathbf{P}_{\mathcal{T}^{k + 1}}
 \mathbf{r}^k\|_2^2 \geq \|\mathbf{P}_{t^{k + 1}}
 \mathbf{r}^k\|_2^2. \nonumber
\end{eqnarray}

Further, by noting that $\|\mathbf{r}^k - \mathbf{r}^{k + 1}\|_2^2 =
\|\mathbf{r}^k\|_2^2 - \|\mathbf{r}^{k + 1}\|_2^2$, we
have
\begin{eqnarray}
  \|\mathbf{r}^k \|_2^2 - \|\mathbf{r}^{k + 1}\|_2^2 &\geq& \|\mathbf{P}_{t^{k + 1}}
  \mathbf{r}^k\|_2^2 \nonumber \\
&\overset{(a)}{\geq}& \|({\phi}_{t^{k + 1}}^\dag)' {\phi}'_{t^{k + 1}} \mathbf{r}^k\|_2^2 \nonumber \\ 
&\overset{(b)}{=}& ({\phi}'_{t^{k + 1}} \mathbf{r}^k)^2 \|\mathbf{\phi}_{t^{k +
  1}}^\dag\|_2^2 \nonumber \\
&\overset{(c)}{=}& \|\mathbf{\Phi}' \mathbf{r}^k\|_\infty^2 \|\mathbf{\phi}_{t^{k + 1}}^\dag \|_2^2  \nonumber \\
&\overset{(d)}{\geq}& {\|\mathbf{\Phi}' \mathbf{r}^k\|_\infty^2} \left( \frac{\|\mathbf{\phi}_{t^{k + 1}}^\dag \mathbf{\phi}_{t^{k + 1}}\|_2^2}{\|\mathbf{\phi}_{t^{k + 1}}\|_2^2}\right)  \nonumber \\
&\overset{(e)}{\geq}&  \frac{\|\mathbf{\Phi}' \mathbf{r}^k\|_\infty^2}{1 + \delta_1}, \label{eq:mmmss4}
\end{eqnarray}
where (a) is because $\mathbf{P}_{t^{k + 1}} = \mathbf{P}'_{t^{k + 1}} = (\mathbf{\phi}_{t^{k + 1}}^\dag)' \mathbf{\phi}'_{t^{k + 1}}$, (b) is from that ${\phi}'_{t^{k + 1}}
\mathbf{r}^k$ is a scalar, (c) is due to the identification rule of OMP, (d) is from the norm inequality, and (e) follows from the RIP.

\end{proof}

\section{Proof of \eqref{eq:sec0}} \label{app:sec0}

\begin{proof}
The main goal of the proof is to establish a lower bound on $\|\mathbf{\Phi}' \mathbf{r}^k\|_\infty^2$.
\vspace{2mm}

Observing that $$(\mathbf{\Phi}' \mathbf{r}^k)_{\mathcal{T}^k} = \mathbf{\Phi}'_{\mathcal{T}^k} (\mathbf{P}^\bot_{\mathcal{T}^k} \mathbf{y})
    = (\mathbf{P}^\bot_{\mathcal{T}^k} \mathbf{\Phi}_{\mathcal{T}^k})' \mathbf{y} = \mathbf{0},$$
we have\begin{equation}
 \|\mathbf{\Phi}' \mathbf{r}^k\|_\infty^2 = \|(\mathbf{\Phi}' \mathbf{r}^k)_{\Omega \backslash \mathcal{T}^k}\|_\infty^2
\end{equation}
Then it follows from the H\"{o}lder's inequality that for all $\mathbf{w} \in \mathcal{R}^n$,
    \begin{eqnarray}
      \|\mathbf{\Phi}' \mathbf{r}^k\|_\infty^2
&{\geq}& \frac{\langle (\mathbf{\Phi}' \mathbf{r}^k)_{\Omega \backslash \mathcal{T}^k}, \mathbf{w}_{\Omega \backslash \mathcal{T}^k} \rangle^2}{\|\mathbf{w}_{\Omega \backslash \mathcal{T}^k}\|_1^2} \nonumber \\
&=& \frac{\langle\mathbf{\Phi}'\mathbf{r}^k, \mathbf{w}\rangle^2}{\|\mathbf{w}_{\Omega \backslash \mathcal{T}^k}\|_1^2} \nonumber \\
&=& \frac{\langle \mathbf{r}^k, \mathbf{\Phi} \mathbf{w}\rangle^2}{\|\mathbf{w}_{\Omega \backslash \mathcal{T}^k}\|_1^2} \label{eq:jiake}
    \end{eqnarray}
Let $\mathbf{w} \in
\mathcal{R}^n$ be a vector such that
\begin{equation} \label{eq:jjjjffff0000}
\mathbf{w}_\mathcal{S} =
\begin{cases}
\mathbf{x}_{\mathcal{S}} & \mathcal{S} \subseteq T \cap \mathcal{T}^{k} \cup
{\Gamma}^k_\tau,\\
\mathbf{0}         & \mathcal{S} \subseteq \Omega  \backslash (T\cap \mathcal{T}^{k} \cup
{\Gamma}^k_\tau).
\end{cases}
\end{equation}
where $\tau = \delta^{1/2}$.
See Figure \ref{fig:supp_w} for an illustration of $\textit{supp}(\mathbf{w})$.
\begin{figure}[t]
\begin{center}
\includegraphics[scale = 1]{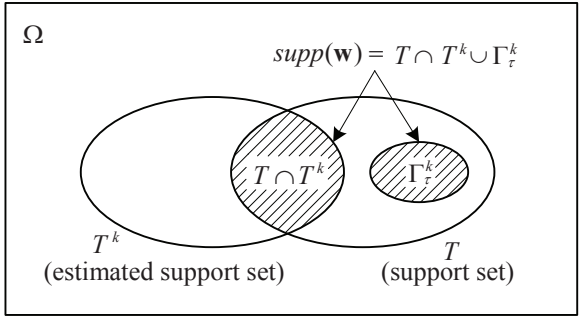}
\caption{Illustration of $\textit{supp}(\mathbf{w})$.} \label{fig:supp_w}
\end{center}
\end{figure}
Then \eqref{eq:jiake} can be rewritten as
    \begin{eqnarray}
      \|\mathbf{\Phi}' \mathbf{r}^k\|_\infty^2 &\geq& \frac{\langle\mathbf{\Phi}'\mathbf{r}^k, \mathbf{w}\rangle^2}{\|\mathbf{x}_{{\Gamma}^k_\tau}\|_1^2}  \nonumber \\
&\geq& \frac{\langle\mathbf{r}^k, \mathbf{\Phi} \mathbf{w}\rangle^2}{\|\mathbf{x}_{{\Gamma}^k_\tau}\|_1^2}  \nonumber \\
&\overset{(a)}{\geq}&  \frac{\langle \mathbf{r}^k, \mathbf{\Phi} \mathbf{w}\rangle^2}{{|{\Gamma}^k_\tau
|} \|\mathbf{x}_{{\Gamma}^k_\tau}\|_2^2}  \nonumber \\
&=&  \frac{\langle \mathbf{r}^k, \mathbf{\Phi} \mathbf{w}\rangle^2}{ \lceil \tau K \rceil\|\mathbf{x}_{{\Gamma}^k_\tau}\|_2^2} \nonumber \\
&\overset{(b)}{=}&  \frac{\langle \mathbf{r}^k, \mathbf{\Phi} (\mathbf{w} - \mathbf{x}^k) \rangle^2}{ \lceil \tau K \rceil\|\mathbf{x}_{{\Gamma}^k_\tau}\|_2^2} \label{eq:jiake1}
    \end{eqnarray}
where (a) follows from the norm inequality
($\|\mathbf{u}\|_1 \leq \sqrt{\|\mathbf{u}\|_0} ~
\|\mathbf{u}\|_2$) and (b) is because $$\langle \mathbf{r}^k, \mathbf{\Phi} \mathbf{x}^k \rangle = \langle \mathbf{P}^\bot_{\mathcal{T}^k} \mathbf{y}, \mathbf{P}_{\mathcal{T}^{k}} \mathbf{y} \rangle = \langle  \mathbf{y}, \mathbf{P}^\bot_{\mathcal{T}^k} \mathbf{P}_{\mathcal{T}^{k}} \mathbf{y} \rangle = \mathbf{0}.$$

We now consider the term $\langle
\mathbf{r}^k, \mathbf{\Phi} (\mathbf{w} - \mathbf{x}^k) \rangle^2$ in the numerator of the right-hand-side of \eqref{eq:jiake1}.
Since the residual can be expressed as $\mathbf{r}^k = \mathbf{\Phi} (\mathbf{x} - \mathbf{x}^k) + \mathbf{v}$, we have
\begin{eqnarray}
\lefteqn{\langle
\mathbf{r}^k, \mathbf{\Phi} (\mathbf{w} - \mathbf{x}^k) \rangle^2} \nonumber \\
 &=& \langle \mathbf{\Phi}(\mathbf{x} - \mathbf{x}^k) + \mathbf{v}, \mathbf{\Phi} (\mathbf{w} - \mathbf{x}^k) \rangle^2 \nonumber\\
 &=& ( \langle \mathbf{\Phi}(\mathbf{x} - \mathbf{x}^k), \mathbf{\Phi} (\mathbf{w} - \mathbf{x}^k) \rangle + \langle \mathbf{v}, \mathbf{\Phi} (\mathbf{w} - \mathbf{x}^k) \rangle )^2 \nonumber\\
 &\overset{(a)}{\geq}& (1 - \tau) \langle \mathbf{\Phi}(\mathbf{x} - \mathbf{x}^k), \mathbf{\Phi} (\mathbf{w} - \mathbf{x}^k) \rangle^2 -  \left( {1}/{\tau} - 1\right) \nonumber \\
  && \times \langle \mathbf{v}, \mathbf{\Phi} (\mathbf{w} - \mathbf{x}^k) \rangle^2 \nonumber\\
&\overset{(b)}{\geq}& (1 - \tau) \langle \mathbf{\Phi}(\mathbf{x} - \mathbf{x}^k), \mathbf{\Phi} (\mathbf{w} - \mathbf{x}^k) \rangle^2 -  \left( {1}/{\tau} - 1\right) \nonumber \\
  && \times \|\mathbf{v}\|_2^2 ~\|\mathbf{\Phi}(\mathbf{w} - \mathbf{x}^k)\|_2^2 \nonumber \\
&\overset{(c)}{=}& \frac{1 - \tau}{4} \left( \|\mathbf{\Phi}(\mathbf{x} - \mathbf{x}^k)\|_2^2 - \|\mathbf{\Phi} (\mathbf{w} - \mathbf{x})\|_2^2 \right. \nonumber \\
  && \left. + \|\mathbf{\Phi} (\mathbf{w} - \mathbf{x}^k)\|_2^2 \right)^2 - \left( {1}/{\tau} - 1\right) \|\mathbf{v}\|_2^2 ~\|\mathbf{\Phi}(\mathbf{w} - \mathbf{x}^k)\|_2^2, \nonumber \\ \label{eq:11me}
\end{eqnarray}
where (a) uses the inequality $$(u + v)^2 \geq (1 - \tau) u^2 - \left( {1}/{\tau} - 1\right) v^2$$
with $u = \langle \mathbf{\Phi}(\mathbf{x} - \mathbf{x}^k), \mathbf{\Phi} (\mathbf{w} - \mathbf{x}^k) \rangle$ and $v = \langle \mathbf{v}, \mathbf{\Phi} (\mathbf{w} - \mathbf{x}^k) \rangle$, (b) is due to the Cauchy-Schwarz inequality, and (c) is from the fact that $$\langle \mathbf{u}, \mathbf{z} \rangle = \frac{1}{2}\left(\|\mathbf{u}\|_2^2 - \|\mathbf{u} - \mathbf{z}\|_2^2 + \|\mathbf{z}\|_2^2\right)$$
with $\mathbf{u} = \mathbf{\Phi}(\mathbf{x} - \mathbf{x}^k)$ and $\mathbf{z} = \mathbf{\Phi} (\mathbf{w} - \mathbf{x}^k)$. 
\vspace{2mm}

Further, observe that
\begin{eqnarray}
\lefteqn{\|\mathbf{\Phi}(\mathbf{x} - \mathbf{x}^k)\|_2^2 - \|\mathbf{\Phi} (\mathbf{w} - \mathbf{x})\|_2^2} \nonumber \\
 & \overset{(a)}{=} & \|\mathbf{\Phi}(\mathbf{x} - \mathbf{x}^k)\|_2^2 - \|\mathbf{\Phi}_{\Gamma^k \backslash \Gamma_{\tau}^k} \mathbf{x}_{\Gamma^k \backslash \Gamma_{\tau}^k}\|_2^2
 \nonumber \\
 & \overset{(b)}{\geq} & (1 - \delta) \|\mathbf{x} - \mathbf{x}^k\|_2^2 - (1 + \delta) \| \mathbf{x}_{\Gamma^k \backslash \Gamma_{\tau}^k}\|_2^2  \nonumber \\
 & {\geq} & (1 - \delta) \|(\mathbf{x} - \mathbf{x}^k)_{\Gamma^k}\|_2^2 - (1 + \delta) \| \mathbf{x}_{\Gamma^k \backslash \Gamma_{\tau}^k}\|_2^2  \nonumber \\
 & \overset{(c)}{=} & (1 - \delta) \|\mathbf{x}_{\Gamma^k}\|_2^2 - (1 + \delta) \| \mathbf{x}_{\Gamma^k \backslash \Gamma_{\tau}^k}\|_2^2  \nonumber \\
 & {=} & (1 - \delta) \left( \|\mathbf{x}_{\Gamma^k_\tau}\|_2^2 + \|\mathbf{x}_{\Gamma^k \backslash \Gamma_{\tau}^k}\|_2^2\right) \nonumber \\
 && - (1 + \delta) \| \mathbf{x}_{\Gamma^k \backslash \Gamma_{\tau}^k}\|_2^2  \nonumber \\
 & = & (1 - \delta) \|\mathbf{x}_{\Gamma^k_\tau}\|_2^2 - 2 \delta \| \mathbf{x}_{\Gamma^k \backslash \Gamma_{\tau}^k}\|_2^2  \nonumber \\
 & \overset{(d)}{\geq} & (1 - \delta) \lceil \tau K \rceil (x_\tau^k)^2 - 2 \delta K ({x} ^k_{\tau})^2  \nonumber \\ 
 & \geq & \left(1 - \delta  - \frac{2 \delta}{\tau}  \right) \lceil \tau K \rceil ({x} ^k_{\tau})^2 \nonumber \\
 & \overset{(e)}{\geq} & (1 - 3\tau ) \lceil \tau K \rceil ({x} ^k_{\tau})^2, \label{eq:12me}
\end{eqnarray}
where (a) is from \eqref{eq:jjjjffff0000}, (b) is due to the RIP, (c) is because $\mathbf{x}^k_{\Gamma^k} = \mathbf{0}$, (d) is because $x_\tau^k$ is the $\lceil \tau K \rceil$-th largest element (in magnitude) in $\mathbf{x}_{\Gamma^k}$ and hence is the smallest one in $\mathbf{x}_{\Gamma_\tau^k}$, and (e) is from that $$1 - \delta  - \frac{2 \delta}{\tau} = 1 - \delta - 2 \tau \geq 1 - 3 \tau.$$

Using \eqref{eq:11me} and \eqref{eq:12me}, we have
\begin{eqnarray}
\lefteqn{\langle \mathbf{r}^k, \mathbf{\Phi} (\mathbf{w} - \mathbf{x}^k) \rangle^2} \nonumber \\
 & {\geq} & \frac{1 - \tau}{4} \left( (1 - 3\tau ) \lceil \tau K \rceil  ({x} ^k_{\tau})^2 + \|\mathbf{\Phi} (\mathbf{w} - \mathbf{x}^k)\|_2^2 \right)^2 \nonumber \\
 && - \left( {1}/{\tau} - 1\right) \|\mathbf{v}\|_2^2 \|\mathbf{\Phi}(\mathbf{w} - \mathbf{x}^k)\|_2^2 \nonumber \\
 & \overset{(a)}{\geq} & (1 - \tau) \left( (1 - 3\tau ) \lceil \tau K \rceil  ({x} ^k_{\tau})^2 \right) \|\mathbf{\Phi} (\mathbf{w} - \mathbf{x}^k)\|_2^2 \nonumber \\
 && - \left( {1}/{\tau} - 1\right) \|\mathbf{v}\|_2^2 \|\mathbf{\Phi}(\mathbf{w} - \mathbf{x}^k)\|_2^2 \nonumber \\
 & \overset{(b)}{\geq} & \left((1 - 4 \tau ) \lceil \tau K \rceil ({x} ^k_{\tau})^2 - \left( {1}/{\tau} - 1\right) \|\mathbf{v}\|_2^2 \right)  \nonumber \\
 && \times \|\mathbf{\Phi}(\mathbf{w} - \mathbf{x}^k)\|_2^2 \nonumber \\
 & \overset{(c)}{\geq} & (1 - 5\tau ) \lceil \tau K \rceil  ({x} ^k_{\tau})^2 \|\mathbf{\Phi}(\mathbf{w} - \mathbf{x}^k)\|_2^2 \nonumber \\
 & \overset{(d)}{\geq} & (1 - 5\tau ) \lceil \tau K \rceil  ({x} ^k_{\tau})^2 (1 - \delta) \|\mathbf{w} - \mathbf{x}^k\|_2^2 \nonumber \\
 & \overset{(e)}{\geq} & (1 - 6\tau ) \lceil \tau K \rceil  ({x} ^k_{\tau})^2 \|\mathbf{w} - \mathbf{x}^k\|_2^2 \nonumber \\
 & {\geq} & (1 - 6\tau ) \lceil \tau K \rceil  ({x} ^k_{\tau})^2 \|(\mathbf{w} - \mathbf{x}^k)_{\Omega \backslash \mathcal{T}^k}\|_2^2 \nonumber \\
 & = & (1 - 6\tau ) \lceil \tau K \rceil ({x} ^k_{\tau})^2 \|\mathbf{x}_{\Gamma^k_\tau}\|_2^2 \label{eq:13me}
\end{eqnarray}
where (a) follows from the fact that $(u + v)^2 \geq 4 u v$ with $u = (1 - 3\tau ) \lceil \tau K \rceil  ({x} ^k_{\tau})^2$ and $v = \|\mathbf{\Phi} (\mathbf{w} - \mathbf{x}^k)\|_2^2$, (b) is because $$(1 - \tau) ( 1 - 3 \tau) = 1 - 4 \tau + 3 \tau^2 \geq 1 - 4 \tau,$$  (c) is because 1) the assumption $\text{SNR} \geq \kappa^2 \delta^{-3/2}$ implies that \begin{eqnarray}
\|\mathbf{\Phi x}\|_2^2  & = & \text{SNR} \cdot {\|\mathbf{v}\|_2^2} \nonumber \\
&\geq &  \frac{\kappa^2}{\delta^{3/2}} \|\mathbf{v}\|_2^2 \nonumber \\
& = &  \frac{\kappa^2}{\tau^3} \|\mathbf{v}\|_2^2 \nonumber \\
& \geq & \frac{\kappa^2}{\tau^3} \left( 1 - \tau \left(\tau - \frac{1}{2}\right)^2 - \frac{3 \tau}{4} \right) \|\mathbf{v}\|_2^2 \nonumber \\
& = & \frac{\kappa^2 (1 + \tau^2)(1 - \tau) \|\mathbf{v}\|_2^2 }{\tau^3},
\end{eqnarray}
and 2) on the other hand,
\begin{eqnarray}
\|\mathbf{\Phi x}\|_2^2  & \overset{\text{RIP}}{\leq} & (1 + \delta) {\|\mathbf{x}\|_2^2} \nonumber \\
& {\leq} & (1 + \delta) K x^2_{\max} \nonumber \\
& {\leq} & (1 + \delta) K (\kappa x_{\min})^2 \nonumber \\
& {\leq} & (1 + \delta) K ( \kappa x^k_{\tau})^2 \nonumber \\
& = & (1 + \tau^2) K ( \kappa x^k_{\tau})^2,
\end{eqnarray}
and hence
\begin{equation}
\left( {1}/{\tau} - 1\right) \|\mathbf{v}\|_2^2 \leq \tau^2 K (x^k_{\tau})^2 \leq \tau \lceil \tau K \rceil (x^k_{\tau})^2, \label{eq:youcai}
\end{equation}
(d) is from the RIP, and (e) follows from that $$(1 - 5 \tau)(1 - \delta) \geq (1 - 5 \tau)(1 - \tau) = 1 - 6 \tau + 5 \tau^2 \geq 1 - 6 \tau.$$

Finally, using \eqref{eq:jiake1} and \eqref{eq:13me}, we have
\begin{eqnarray}
\|\mathbf{\Phi}' \mathbf{r}^k\|_\infty^2 &\geq& \frac{\langle \mathbf{r}^k, \mathbf{\Phi} (\mathbf{w} - \mathbf{x}^k) \rangle^2}{ \lceil \tau K \rceil\|\mathbf{x}_{{\Gamma}^k_\tau}\|_2^2} \nonumber \\
& \geq &
\frac{(1 - 6\tau ) \lceil \tau K \rceil ({x} ^k_{\tau})^2 \|\mathbf{x}_{\Gamma^k_\tau}\|_2^2}{ \lceil \tau K \rceil\|\mathbf{x}_{{\Gamma}^k_\tau}\|_2^2} \nonumber \\
& = & (1 - 6 \tau) ({x} ^k_{\tau})^2, \label{eq:mmmss4422}
\end{eqnarray}
which completes the proof.
\end{proof}

\bibliographystyle{IEEEbib}
\bibliography{CS_refs}

\begin{thebibliography}{10}

\bibitem{donoho2006compressed}
D.~L. Donoho,
\newblock ``Compressed sensing,''
\newblock {\em IEEE Trans. Inform. Theory}, vol. 52, no. 4, pp. 1289--1306,
  Apr. 2006.

\bibitem{candes2006robust}
E.~J. Cand{\`e}s, J.~Romberg, and T.~Tao,
\newblock ``{Robust uncertainty principles: Exact signal reconstruction from
  highly incomplete frequency information},''
\newblock {\em IEEE Trans. Inform. Theory}, vol. 52, no. 2, pp. 489--509, Feb.
  2006.

\bibitem{devore1993constructive}
R.~A. DeVore and G.~G. Lorentz,
\newblock {\em Constructive approximation}, vol. 303,
\newblock Springer, 1993.

\bibitem{chen1999atomic}
S.~S. Chen, D.~L. Donoho, and M.~A. Saunders,
\newblock ``{Atomic decomposition by basis pursuit},''
\newblock {\em SIAM journal on scientific computing}, vol. 20, no. 1, pp.
  33--61, 1999.

\bibitem{schaefer1992subset}
R.~L. Schaefer,
\newblock ``Subset selection in regression,''
\newblock {\em Technometrics}, vol. 34, no. 2, pp. 229--229, 1992.

\bibitem{natarajan1995sparse}
B.~K. Natarajan,
\newblock ``Sparse approximate solutions to linear systems,''
\newblock {\em SIAM journal on computing}, vol. 24, no. 2, pp. 227--234, Apr.
  1995.

\bibitem{pati1993orthogonal}
Y.~C. Pati, R.~Rezaiifar, and P.~S. Krishnaprasad,
\newblock ``Orthogonal matching pursuit: Recursive function approximation with
  applications to wavelet decomposition,''
\newblock in {\em Proc. 27th Annu. Asilomar Conf. Signals, Systems, and
  Computers}. IEEE, Nov. Pacific Grove, CA, Nov. 1993, vol.~1, pp. 40--44.

\bibitem{mallat1993matching}
S.~G. Mallat and Z.~Zhang,
\newblock ``Matching pursuits with time-frequency dictionaries,''
\newblock {\em IEEE Trans. Signal Process.}, vol. 41, no. 12, pp. 3397--3415,
  Dec. 1993.

\bibitem{zhang2009consistency}
T.~Zhang,
\newblock ``On the consistency of feature selection using greedy least squares
  regression,''
\newblock {\em J. of Mach. Learn. Res.}, vol. 10, pp. 555--568, 2009.

\bibitem{tropp2010computational}
J.~A. Tropp and S.~J. Wright,
\newblock ``Computational methods for sparse solution of linear inverse
  problems,''
\newblock {\em Proceedings of the IEEE}, vol. 98, no. 6, pp. 948--958, Jun.
  2010.

\bibitem{tropp2007signal}
J.~A. Tropp and A.~C. Gilbert,
\newblock ``{Signal recovery from random measurements via orthogonal matching
  pursuit},''
\newblock {\em IEEE Trans. Inform. Theory}, vol. 53, no. 12, pp. 4655--4666,
  Dec. 2007.

\bibitem{fletcher2012orthogonal}
A.~K. Fletcher and S.~Rangan,
\newblock ``Orthogonal matching pursuit: A brownian motion analysis,''
\newblock {\em IEEE Trans. Signal Process.}, vol. 60, no. 3, pp. 1010--1021,
  Mar. 2012.

\bibitem{tropp2004greed}
J.~A. Tropp,
\newblock ``{Greed is good: Algorithmic results for sparse approximation},''
\newblock {\em IEEE Trans. Inform. Theory}, vol. 50, no. 10, pp. 2231--2242,
  Oct. 2004.

\bibitem{davenport2010analysis}
M.~A. Davenport and M.~B. Wakin,
\newblock ``{Analysis of Orthogonal Matching Pursuit using the restricted
  isometry property},''
\newblock {\em IEEE Trans. Inform. Theory}, vol. 56, no. 9, pp. 4395--4401,
  Sep. 2010.

\bibitem{liu2012orthogonal}
E.~Liu and V.~N. Temlyakov,
\newblock ``The orthogonal super greedy algorithm and applications in
  compressed sensing,''
\newblock {\em IEEE Trans. Inform. Theory}, vol. 58, no. 4, pp. 2040--2047,
  Apr. 2012.

\bibitem{huang2011recovery}
S.~Huang and J.~Zhu,
\newblock ``Recovery of sparse signals using \text{OMP} and its variants:
  convergence analysis based on \text{RIP},''
\newblock {\em Inverse Problems}, vol. 27, no. 3, pp. 035003, 2011.

\bibitem{zhang2011sparse}
T.~Zhang,
\newblock ``Sparse recovery with orthogonal matching pursuit under
  \text{RIP},''
\newblock {\em IEEE Trans. Inform. Theory}, vol. 57, no. 9, pp. 6215--6221,
  Sep. 2011.

\bibitem{livshits2012efficiency}
E.~D. Livshits,
\newblock ``On the efficiency of the orthogonal matching pursuit in compressed
  sensing,''
\newblock {\em Sbornik: Mathematics}, vol. 203, no. 2, pp. 183, 2012.

\bibitem{wang2012Recovery}
J.~Wang and B.~Shim,
\newblock ``On the recovery limit of sparse signals using orthogonal matching
  pursuit,''
\newblock {\em IEEE Trans. Signal Process.}, vol. 60, no. 9, pp. 4973--4976,
  Sep. 2012.

\bibitem{mo2012remarks}
Q.~Mo and Y.~Shen,
\newblock ``A remark on the restricted isometry property in orthogonal matching
  pursuit algorithm,''
\newblock {\em IEEE Trans. Inform. Theory}, vol. 58, no. 6, pp. 3654--3656,
  Jun. 2012.

\bibitem{soussen2013joint}
C.~Soussen, R.~Gribonval, J.~Idier, and C.~Herzet,
\newblock ``Joint $k$-step analysis of orthogonal matching pursuit and
  orthogonal least squares,''
\newblock {\em IEEE Trans. Inform. Theory}, vol. 59, no. 5, pp. 3158--3174, May
  2013.

\bibitem{donoho2001uncertainty}
D.~L. Donoho and X.~Huo,
\newblock ``Uncertainty principles and ideal atomic decomposition,''
\newblock {\em IEEE Trans. Inform. Theory}, vol. 47, no. 7, pp. 2845--2862,
  Nov. 2001.

\bibitem{candes2005decoding}
E.~J. Cand{\`e}s and T.~Tao,
\newblock ``{Decoding by linear programming},''
\newblock {\em IEEE Trans. Inform. Theory}, vol. 51, no. 12, pp. 4203--4215,
  Dec. 2005.

\bibitem{wen2013improved}
J.~Wen, X.~Zhu, and D.~Li,
\newblock ``Improved bounds on restricted isometry constant for orthogonal
  matching pursuit,''
\newblock {\em Electronics Letters}, vol. 49, no. 23, pp. 1487--1489, 2013.

\bibitem{chang2014improved}
L.~Chang and J.~Wu,
\newblock ``An improved \text{RIP}-based performance guarantee for sparse
  signal recovery via orthogonal matching pursuit,''
\newblock {\em IEEE Trans. Inform. Theory}, vol. 60, no. 9, pp. 5702--5715,
  Sep. 2014.

\bibitem{cai2011orthogonal}
T.~T. Cai and L.~Wang,
\newblock ``Orthogonal matching pursuit for sparse signal recovery with
  noise,''
\newblock {\em IEEE Trans. Inform. Theory}, vol. 57, no. 7, pp. 4680--4688,
  Jul. 2011.

\bibitem{shen2011sparse}
Y.~Shen and S.~Li,
\newblock ``Sparse signals recovery from noisy measurements by orthogonal
  matching pursuit,''
\newblock {\em arXiv:1105.6177}, 2011.

\bibitem{wu2013exact}
R.~Wu, W.~Huang, and D~Chen,
\newblock ``The exact support recovery of sparse signals with noise via
  orthogonal matching pursuit,''
\newblock {\em IEEE Signal Processing Letters}, vol. 20, no. 4, pp. 403--406,
  Apr. 2013.

\bibitem{needell2009cosamp}
D.~Needell and J.~A. Tropp,
\newblock ``\text{CoSaMP}: Iterative signal recovery from incomplete and
  inaccurate samples,''
\newblock {\em Applied and Computational Harmonic Analysis}, vol. 26, no. 3,
  pp. 301--321, Mar. 2009.

\bibitem{kwon2013multipath}
S.~Kwon, J.~Wang, and B.~Shim,
\newblock ``Multipath matching pursuit,''
\newblock {\em IEEE Trans. Inform. Theory}, vol. 60, no. 5, pp. 2986--3001, May
  2014.

\bibitem{candes2008restricted}
E.~J. Cand{\`e}s,
\newblock ``{The restricted isometry property and its implications for
  compressed sensing},''
\newblock {\em Comptes Rendus Mathematique}, vol. 346, no. 9-10, pp. 589--592,
  2008.

\bibitem{shen2014analysis}
Y.~Shen, B.~Li, W.~Pan, and J.~Li,
\newblock ``Analysis of generalised orthogonal matching pursuit using
  restricted isometry constant,''
\newblock {\em Electronics Letters}, vol. 50, no. 14, pp. 1020--1022, 2014.

\bibitem{herzet2012exact}
C.~Herzet, C.~Soussen, J.~Idier, and R.~Gribonval,
\newblock ``Exact recovery conditions for sparse representations with partial
  support information,''
\newblock {\em IEEE Trans. Inform. Theory}, vol. 59, no. 11, pp. 7509--7524,
  Nov. 2013.

\bibitem{foucart2013stability}
S.~Foucart,
\newblock ``Stability and robustness of weak orthogonal matching pursuits,''
\newblock in {\em Recent Advances in Harmonic Analysis and Applications}, pp.
  395--405. Springer, 2013.

\bibitem{wang2012how}
J.~Wang and B.~Shim,
\newblock ``How many iterations are needed for the exact recovery of sparse
  signals using orthogonal matching pursuit,''
\newblock {\em arXiv:1211.4293v2}, 2012.

\bibitem{livshitz2014sparse}
E.~D. Livshitz and V.~N. Temlyakov,
\newblock ``Sparse approximation and recovery by greedy algorithms,''
\newblock {\em IEEE Trans. Inform. Theory}, vol. 60, no. 7, pp. 3989--4000,
  Jul. 2014.

\bibitem{needell2010signal}
D.~Needell and R.~Vershynin,
\newblock ``{Signal recovery from incomplete and inaccurate measurements via
  regularized orthogonal matching pursuit},''
\newblock {\em IEEE J. Sel. Topics Signal Process.}, vol. 4, no. 2, pp.
  310--316, Apr. 2010.

\bibitem{donoho2006sparse}
D.~L. Donoho, I.~Drori, Y.~Tsaig, and J.~L. Starck,
\newblock ``{Sparse solution of underdetermined linear equations by stagewise
  orthogonal matching pursuit},''
\newblock {\em IEEE Trans. Inform. Theory}, vol. 58, no. 2, pp. 1094--1121,
  Feb. 2012.

\bibitem{dai2009subspace}
W.~Dai and O.~Milenkovic,
\newblock ``{Subspace pursuit for compressive sensing signal reconstruction},''
\newblock {\em IEEE Trans. Inform. Theory}, vol. 55, no. 5, pp. 2230--2249, May
  2009.

\bibitem{blumensath2009iterative}
T.~Blumensath and M.~E. Davies,
\newblock ``Iterative hard thresholding for compressed sensing,''
\newblock {\em Applied and Computational Harmonic Analysis}, vol. 27, no. 3,
  pp. 265--274, 2009.

\bibitem{foucart2011hard}
S.~Foucart,
\newblock ``Hard thresholding pursuit: an algorithm for compressive sensing,''
\newblock {\em SIAM Journal on Numerical Analysis}, vol. 49, no. 6, pp.
  2543--2563, 2011.

\bibitem{wang2012Generalized}
J.~Wang, S.~Kwon, and B.~Shim,
\newblock ``Generalized orthogonal matching pursuit,''
\newblock {\em IEEE Trans. Signal Process.}, vol. 60, no. 12, pp. 6202--6216,
  Dec. 2012.

\bibitem{baraniuk2008simple}
R.~Baraniuk, M.~Davenport, R.~DeVore, and M.~Wakin,
\newblock ``{A simple proof of the restricted isometry property for random
  matrices},''
\newblock {\em Constructive Approximation}, vol. 28, no. 3, pp. 253--263, 2008.

\end{thebibliography}
\end{document}